\documentclass[tikz,reqno,11pt]{amsart}
\usepackage[linktocpage]{hyperref}
\usepackage[utf8]{inputenc}
\usepackage{graphicx}
\usepackage{slashed}
\usepackage{amsmath, amscd}
\usepackage[mathscr]{eucal}
\usepackage{amssymb}
\usepackage{esint}
\usepackage{enumitem}
\usepackage{mathtools}
\usepackage{tensor}
\usepackage{comment}
\usepackage{bbm}
\usepackage{dsfont}
\usepackage{ifpdf}
\usepackage{tikz}

\ifpdf\else\usepackage{pstricks}\fi

\usepackage{pst-coil}
\usepackage{pst-grad}
\usepackage{pst-node, pst-plot}
\usepackage{epsfig}
\usepackage{pstricks-add}
\usepackage[nomessages]{fp}
\usepackage[space]{grffile} % For spaces in paths
\usepackage{etoolbox} % For spaces in paths

\makeatletter % For spaces in paths
\patchcmd\Gread@eps{\@inputcheck#1 }{\@inputcheck"#1"\relax}{}{}
\makeatother

\numberwithin{equation}{section}
\allowdisplaybreaks[1]

\newtheorem{Def}{Definition}[section]

\newtheorem{Thm}[Def]{Theorem}
\newtheorem{Prp}[Def]{Proposition}
\newtheorem{Lemma}[Def]{Lemma}
\newtheorem{Remark}[Def]{Remark}
\newtheorem{Corollary}[Def]{Corollary}

\newcommand{\Thanks}{\vspace*{.5em} \noindent \thanks}
\newcommand{\beq}{\begin{equation}}
\newcommand{\eeq}{\end{equation}}

\newcommand{\Proof}{\begin{proof}}
\newcommand{\QED}{\end{proof} \noindent}
\newcommand{\QEDrem}{\ \hfill $\Diamond$}

\newcommand{\la}{\langle}
\newcommand{\ra}{\rangle}

\newcommand{\braket}[2]{\Big\langle #1 \Big| #2 \Big\rangle}
\newcommand{\Scpr}[2]{(\, #1 \, | \, #2\,)}

\newcommand{\jh}[1]{\widehat{\mathcal{J}}_{#1}}
\newcommand{\jc}[1]{\widecheck{\mathcal{J}}_{#1}}
\newcommand{\Sl}{\mathopen{\prec}}
\newcommand{\Sr}{\mathclose{\succ}}

\newcommand{\C}{\mathbb{C}}
\newcommand{\R}{\mathbb{R}}

\newcommand{\Z}{\mathbb{Z}}
\newcommand{\N}{\mathbb{N}}

\newcommand{\mm}{\mycal M}
\newcommand{\nn}{\mycal N}
\newcommand{\Am}{\mathcal{A}}
\newcommand{\hh}{\mathcal{H}}

\newcommand{\bb}{{\mathcal{B}}}

\newcommand{\rr}{\mathcal{R}}

\newcommand{\Dir}{{\mathcal{D}}}
\newcommand{\D}{{\mathscr{D}}}

\newcommand{\Cisc}{C^\infty_{{\rm{sc}}}}
\newcommand{\Cisco}{C^\infty_{{\rm{sc}},0}}

\newcommand{\Idmat}{\mathbbmss{1}}

\newcommand{\dif}[1]{\text{\rm{d}}#1}

\newcommand{\PP}{\mathfrak{p}}

\newcommand{\OmegaVar}{\omega_{\varepsilon}}
\newcommand{\V}{\noindent}

\newcommand{\pv}[1]{{\rm{PV}}\Bigg(#1\Bigg)}

\newcommand{\bitem}{\begin{itemize}[leftmargin=2.5em]}
\newcommand{\eitem}{\end{itemize}}

\newcommand{\GL}[2]{{\rm{GL}}_{#1}(#2)}

\renewcommand{\L}{{\mathcal{L}}}

\renewcommand{\S}{\mathscr{S}}

\DeclareMathOperator{\diag}{diag}

\DeclareMathOperator{\supp}{supp}

\newcommand{\scrM}{\mycal M}
\newcommand{\scrN}{\mycal N}

\DeclareFontFamily{OT1}{rsfso}{}
\DeclareFontShape{OT1}{rsfso}{m}{n}{ <-7> rsfso5 <7-10> rsfso7 <10-> rsfso10}{}
\DeclareMathAlphabet{\mycal}{OT1}{rsfso}{m}{n}

\DeclareFontFamily{U}{mathx}{\hyphenchar\font45}
\DeclareFontShape{U}{mathx}{m}{n}{
      <5> <6> <7> <8> <9> <10>
      <10.95> <12> <14.4> <17.28> <20.74> <24.88>
      mathx10
      }{}
\DeclareSymbolFont{mathx}{U}{mathx}{m}{n}
\DeclareFontSubstitution{U}{mathx}{m}{n}
\DeclareMathAccent{\widecheck}{0}{mathx}{"71}
\DeclareMathAccent{\wideparen}{0}{mathx}{"75}

\usetikzlibrary{
  decorations.pathmorphing,
  decorations.markings,
  arrows.meta,
  calc
}

\tikzset{
  scri/.style        = {thick, blue!70!black},
  horizon/.style     = {thick, red!80!black, dashed},
  cauchy-hor/.style  = {thick, purple!80!black, dashed},
  singularity/.style = {
    thick, red,
    decoration={snake, amplitude=2.2pt, segment length=3.5pt, pre length=2pt, post length=2pt},
    decorate
  },
  em-wave/.style     = {
    orange!60!red,
    decoration={snake, amplitude=1.5pt, segment length=1pt},
    decorate
  },
  fill-I/.style      = {fill=gray!12},
  fill-II/.style     = {fill=yellow!18},
  fill-III/.style    = {fill=blue!10},
  fill-IV/.style     = {fill=gray!12},
}

\title[Fermionic signature operator in Reissner--Nordstr\"om Geometry]{The Fermionic
    Signature Operator in the Reissner--Nordstr\"om Geometry in Horizon--Penetrating Coordinates}

\author[F.\ Finster]{Felix Finster}
\author[C.\ Krpoun]{Christoph Krpoun \\ \\ May 2026}
\address{Fakultät f\"ur Mathematik \\ Universit\"at Regensburg \\ D-93040
Regensburg \\ Germany}
\email{finster@ur.de, christoph.krpoun@mathematik.ur.de}

\begin{document}
\maketitle
\begin{abstract}
    We study the Dirac equation in the Reissner--Nordström geometry 
in horizon-penetrating coordinates up to 
    the Cauchy horizon. A mass decomposition theorem is proved,
    which gives a covariant 
    representation of the spacetime inner product that naturally involves the fermionic signature 
    operator and the fermionic flux operator. We compute their spectra and show that both are bounded 
    symmetric operators on the solution space $\hh_m$ of the massive Dirac equation. The 
    corresponding fermionic projector state is constructed and shown to satisfy the Hadamard 
    condition.
%    Furthermore, we construct a thermal $\beta$-KMS state from the fermionic signature 
%    operator for a fixed angular mode $(k,l)$, as seen by an observer at infinity. This state is 
%    likewise Hadamard. By applying a UV-cutoff argument, one may extend the thermal $\beta$-KMS 
%    state to the full set of angular modes, thereby providing a fully covariant construction of 
%    Hawking radiation.
%    
    Lastly, we give some physical interpretations of the fermionic flux operator.
\end{abstract}
\tableofcontents
\section{Introduction}
The fermionic signature operator $\S_m$ was first introduced in~\cite{finite, infinite} and is 
based on earlier perturbative constructions in~\cite{sea}. It provides a framework for spectral 
geometry in Lorentzian signature~\cite{drum} and can be used to define Hadamard states.
These states play an essential role in quantum physics because they provide the physically 
meaningful starting point for perturbative treatments of interacting quantum fields (see for example~\cite
{fewster2013necessity, brunetti+dutsch+fredenhagen, dappiaggiDirac}). The existence of Hadamard 
states was established abstractly in~\cite{fulling+narcowich+wald}, while concrete methods for 
their explicit construction were developed both via differential geometric techniques~\cite
{GeradWrochna2013, GeradWrochna2016, gérard2017hadamardstateslinearizedyangmills} and via 
holographic methods in~\cite{dappiaggi+hack+pinamonti, dappiaggi+moretti, gerard+wrochna2}.
Another method for constructing Hadamard states,
which has the advantage of being fully covariant,
is provided by the fermionic signature operator.
This method has already been successfully applied to the construction of 
quasi-free Dirac states in globally hyperbolic spacetimes~\cite{hadamard, fewster+lang, rindler} and to the 
exterior Schwarzschild geometry in~\cite{sigbh}.

The goal of the present paper is to construct the fermionic
signature operator in the Reissner-Nordström geometry
geometry with horizon-penetrating coordinates up to the
Cauchy horizon. As an application, we construct quasi-free
Dirac states and study their properties.
For the construction of the fermionic signature operator
in black hole spacetimes, the horizons cause major complications.
Namely, in globally spacetimes without horizons, the
the mass oscillation property introduced in~\cite{infinite} guarantees 
well-defined spacetime inner products, ensuring that the definition
of~$\S_m$ is mathematically meaningful. 
However, as already observed in the exterior Schwarzschild geometry
in~\cite{sigbh}, the mass oscillation property fails to hold if
a horizon is present, in simple terms due to the redshift effect.
As a way out, in~\cite{sigbh} a generalized {\em{mass decomposition}}
was introduced, which also tell us about the flux of fermions
through the horizon.

In the present work, we extend this analysis to the Reissner--Nordström geometry with 
horizon-penetrating coordinates, where the situation is further complicated by the presence of the 
Cauchy horizon. Using the integral representation of the Dirac propagator, we obtain an analogous 
mass decomposition (see Theorem~\ref{thm:massDecompTheo}).
\begin{align*}
    \braket{\PP \psi}{\PP \psi} &= \dfrac{1}{2\pi}\int_I ( \psi_m | \S_m  \phi _m)_m \:\dif{m} 
    \notag \\
    &\quad\:+\dfrac{i}{\pi} \int_I \dif{m} \, \int_I \dif{m'}\: \pv{\dfrac{1}{m-m'}} \, \mathfrak{B}
    (\psi_m, \, \phi_{m'}) \: ,
\end{align*}
where $\mathfrak{B}$ is a smooth function of $m$ and $m'$ (see also Remark~\ref{rmk:flux}), and 
PV denotes the Cauchy principal value. The scalar product $\Scpr{\cdot}{\cdot}_m$ refers to 
solutions of the Dirac equation with fixed mass $m$ (as defined in~\eqref{scalPro}). The 
decomposition holds for all families $(\psi_m)_{m\in I}$ and $(\phi_m)_{m\in I}$ in the domain 
$\hh^\infty$ (see Definition~\ref{subsetH}).

The main structural difference compared to the Schwarzschild case is that the fermionic
flux operator $\mathfrak{B}_m$ (see Definition~\ref{def:flux_op}) acquires a modified form due to 
the possibility of Dirac waves crossing the event horizon and propagating to the Cauchy horizon. 
Additionally, the spectra of both $\S_m$ and $\mathfrak{B}_m$ depend on the evolution of the wave 
inside the black hole between the Cauchy and event horizons. The explicit effects can only be 
obtained numerically, since explicit solutions are available only asymptotically (see~\cite
{krpoun-mueller} and~\cite{finster-krpoun1}).

The fermionic signature operator $\S_m$ is a symmetric and bounded operator on the Hilbert space 
$\hh_m$. Being constructed from Dirac wave functions, $\S_m$ reflects structural properties of the 
solution space of the Dirac equation. Its fundamental properties are studied in~\cite{finite, 
infinite, sigsymm}, to which we refer the interested reader. In the following, we state the main 
results of this paper.

The fermionic signature operator is a bounded, symmetric operator on $\hh^\infty_m$ with 
operator norm $\|\S_m\|_{{\rm{op}}} \leq 2$. Its construction is fully covariant. Since it respects 
the symmetries of the spacetime, it admits a joint spectral decomposition with the angular operator 
$\mathcal{A}$ (see~\eqref{matricesDirac}) and the Dirac Hamiltonian $H$ (for details, see 
Corollary~\ref{cor:spec_decomp_sig}). 
\[
    \S_m = \sum_{k,l} \int_{|\omega|>m} S^{kl}_m(\omega) \,\dif{E}_\omega\, A_{k,l} \, .
\]
Interestingly, the eigenvalues of $\S_m$ contain rich information about the black hole geometry. 
Examining their explicit form in Corollary~\ref{cor:props_s}, we see that $\S_m$ vanishes for all 
$\omega \in [-m,m]$. This naturally supports the interpretation that only solutions with sufficient 
energy to escape to infinity contribute, while low-energy solutions fall into the black hole.

Furthermore, $\S_m$ decomposes into positive and negative spectral subspaces corresponding to a 
frequency splitting for an observer at infinity, thereby naturally reproducing the Dirac sea 
picture. We also construct fermionic projector states from $\S_m$ and show that these states 
satisfy the Hadamard condition (Corollary~\ref{projHad}).

%Finally, for an observer at infinity and for a fixed angular mode, we can construct thermal 
%\Felix{Diskutiere das.}%
%$\beta$-KMS states from $\S_m$, which are likewise Hadamard (Theorem~\ref{thm:thermal_state}), 
%strongly suggesting an interpretation in terms of Hawking radiation. We emphasize that the 
%construction of these states is fully 
%independent of an observer. By applying a UV-cutoff argument to the set of possible angular modes, 
%one may extend the thermal $\beta$-KMS state to the full set of physical modes, giving the first 
%fully observer-independent description of Hawking radiation.

Additionally, the fermionic flux operator $\mathfrak{B}_m$ is a bounded symmetric operator on 
$\hh^\infty_m$ with operator norm $\|\mathfrak{B}_m\|_{{\rm{op}}} \leq 1$ (for more information, 
see Proposition~\ref{prp:fluxOp}). Since it is symmetric, it also admits a joint spectral 
decomposition together with $\mathcal{A}$ and $H$.
\[
    \mathfrak{B}_m = \sum_{k,l} \int_{|\omega|>m} \mathfrak{B}^{kl}_m(\omega) \,\dif{E}_\omega\, A_
    {k,l} \, .
\]
Interestingly, even though the full non-singular part of the mass decomposition does depend on the 
Wronskians (for the definition of the Wronskian, see~\ref{def:wronskian}) at the Cauchy and event 
horizons, the fermionic flux operator does not. This is due to the fact that the fermionic flux operator 
measures the flux through the event horizon and the entire wave propagates through the Cauchy 
horizon. Everything that enters the black hole travels to the singularity. Thus, we only have a 
causal boundary at the event horizon, which is manifest in the mass decomposition.

This paper is organized as follows. In Section~\ref{sec:prems}, we provide the theoretical 
background on the Dirac equation in globally hyperbolic spacetimes, with particular focus on the 
Reissner--Nordström geometry. In Section~\ref{sec:mass_decomp}, we establish the mass decomposition 
theorem. Section~\ref{section4} is devoted to the computation of the fermionic signature operator, 
the analysis of its spectral properties, and the construction of quasi-free states. Finally, in 
Section~\ref{section5}, we derive and analyze the fermionic flux operator and give some physical 
interpretations.
\section{Preliminaries \label{sec:prems}}
\subsection{The Dirac Equation in a Globally Hyperbolic Spacetime}
We begin with preliminaries on the Dirac equation in globally hyperbolic spacetimes, following the 
presentation in~\cite{finite, infinite}. Let $(\mm, g)$ be a four-dimensional, smooth, globally 
hyperbolic Lorentzian spin manifold. For the signature of the metric, we adopt the convention $(+,-,
-,-)$. As proven in~\cite{bernal+sanchez}, $\scrM$ admits a smooth foliation $(\scrN_\tau)_{\tau 
\in \R}$ by Cauchy hypersurfaces.

We denote the corresponding spinor bundle by $S\mm$, whose fibers $S_x\mm$ are equipped with an 
inner product of signature $(2,2)$, denoted by $\Sl \cdot | \cdot \Sr_{x}$. The smooth sections of 
the spinor bundle are written as $C^{\infty}(\mm, S\mm)$, while $C_0^{\infty}(\mm, S\mm)$ denotes 
smooth sections with compact support. These sections are also referred to as wave functions and are 
denoted by $\psi$ or $\phi$.

On the space of wave functions, we introduce the Lorentz-invariant inner product
\begin{align}\label{indefitSP}
    \langle \psi | \phi \rangle & : C^{\infty}(\mm, S\mm) \times C_0^{\infty}
    (\mm, S\mm) \longrightarrow \mathbb{C} \notag \\
                                & \langle \psi | \phi \rangle := \int_{\mm}
                                \Sl \psi | \phi \Sr_{x} \, d\mu_{\mm}.
\end{align}
We consider the Dirac equation for a given mass parameter $m \geq 0$, written as
\[ (\mathcal{D} - m )\, \psi_m = 0 \:, \]
where the Dirac operator $\mathcal{D}$ takes the form
\[ 
    \mathcal{D} = iG^{k}\partial_{k} + \bb : C^{\infty}(\mm, S\mm) \longrightarrow C^{\infty}(\mm, 
    S\mm) \:, 
\]
and $G : T_x\mm \longrightarrow L(S_x\mm)$
are the Dirac matrices satisfying the anti-commutation relations:
\[
    G(v)G(u) + G(u)G(v) = 2g(v,u)\,\Idmat_{S_x \mm}. 
\]
This map can be understood as the representation of the Clifford multiplication in terms of the 
components of the general Dirac matrices in curved spacetime.

We employ the Feynman dagger notation $\slashed{\nu} = G^j \nu_{j}$. The connection part of the 
covariant derivative is collected in the term $\bb$. The Dirac operator can equivalently be written 
as $\mathcal{D} = i G^j \nabla_j$, where $\nabla$ denotes the Levi-Civita spin connection on 
$S\mm$. For further details on the Dirac equation in curved spacetimes, we refer to~\cite{finite} 
and~\cite{lawson+michelsohn}.

Given initial data on a Cauchy surface, the Dirac equation admits unique global solutions. For 
compactly supported initial data, finite propagation speed implies that the resulting solution 
also has compact support on any other Cauchy surface. Such solutions are referred to as being \emph
{spatially compact}. We will denote the smooth, spatially compact solutions are by $\Cisc(\mm, S\mm)
$.

For such solutions with fixed mass $m$, the scalar product is defined as
\begin{align}\label{scalPro}
    ( \psi_m | \phi_m)_m = \int_{\nn} \Sl \psi_m | \nu^j G_j \, \phi_m\Sr_x \, \dif{\mu}_{\nn}(x),
\end{align}
where $\nu$ denotes the future-directed normal on $\nn$. Owing to current conservation, this scalar 
product is independent of the choice of $\scrN$ (see~\cite[Section~2]{finite} for details).
Forming the completion gives the Hilbert space $(\hh_m, ( \cdot | \cdot )_m)$.

In this paper, we restrict our attention to \emph{stationary} spacetimes, meaning that there exists 
a Killing field $K$ which is asymptotically timelike (for the general definition, see~\cite
{oneill}). We choose the foliation $(\scrN_\tau)_{\tau \in \R}$ such that the Killing field is 
given by $K = \partial_\tau$. In this setting, it is convenient to represent spacetime as a product 
$\scrM = \R \times \scrN$. The Dirac equation can then be written in the Hamiltonian form
\[
    i\partial_{\tau} \psi = H \psi, \quad \text{with} \quad H := -(G^{\tau})^{-1} \left( i \sum_
    {\alpha=1}^3 G^\alpha \partial_\alpha + \bb - m \right),
\]
where the Hamiltonian $H$ acts on $\hh$ with dense domain $\D(H) = C^\infty_0(\scrN, S\scrN)$.

\subsection{The Mass Oscillation Property}
Since we are considering a spacetime with infinite lifetime, the inner product~\eqref{indefitSP} 
may diverge for Dirac solutions. To obtain meaningful results, we have summarize the mass
oscillation formalism to this end, which has been developed in detail in~\cite{infinite, finite}.

We consider families of Dirac solutions $(\psi_m)_{m \in I}$ with a mass parameter $m$ 
varying in a bounded interval $I := (m_L, m_R)$ such that $0 \notin I$. Given Cauchy surface 
$\nn$, we regard $\psi_{\nn}(x, m)$ as an element of $S_x\mm$ and assume that $\psi_{\nn}(x,m) \in
C_0^{\infty}(\nn \times I, S\mm)$. 

In the following, we denote the mass dependence by a subscript, $\psi_m(x) := \psi(x, m)$. For any 
fixed $m$, the scalar product~\eqref{scalPro} is well defined. Extending this to families of 
solutions $\psi_m, \phi_m \in \Cisco(\mm \times I, S\mm)$, we define a scalar product by 
integrating over the mass parameter:
\[
    (\psi | \phi) := \int_I ( \psi_m | \phi_m )_m \, \dif{m}
\]
(where $\dif{m}$ is the Lebesgue measure). Taking the completion yields the Hilbert space $(\hh, 
(\cdot | \cdot))$, consisting of measurable functions for which, for almost all $m \in I$, the 
function $\psi_m(x)$ is a weak solution of the Dirac equation with $\psi_m|_{\nn} \in L^2(\nn)$. By 
construction, $\Cisco(\mm, S\mm)$ is dense in $\hh$. Moreover, the spatial integral is 
integrable over $m \in I$, ensuring that the above scalar product is well defined. The 
corresponding norm on $\hh$ is denoted by $\|\cdot\|$.

The Hilbert space $\hh$ can be understood as the direct integral of the Hilbert spaces $\hh_m$, 
written as
\[
    \hh = L^2(I, \hh_m \, ; dm) = \int_I^{\oplus} \hh_m \, \dif{m}
\]
On $\hh$, we define an operator that forms superpositions of waves oscillating with different 
frequencies:
\begin{align} \label{oscOp}
    \PP : \hh \longrightarrow \Cisc(\mm,S\mm), \quad \psi \mapsto \int_I \psi_m \, \dif{m} \:.
\end{align}
The oscillatory contributions cancel in the integral, having the
effect that the superposition of Dirac wave functions decays
for large times. This resolves the convergence problems 
of the time integral in~\eqref{indefitSP}. This decay also allows one to 
integrate the Dirac operator by parts within the spacetime inner product~\eqref{indefitSP}, yielding
\[
    \la \PP \mathcal{D} \psi | \PP\phi \ra = \la \PP\psi | \PP \mathcal{D} \phi \ra \: .
\]
We will see later (Section~\ref{sec:mass_decomp}) that this condition does not hold in the presence 
of horizons, making it necessary to generalize the setting.
\subsection{The Dirac Equation in the Reissner-Nordstr\"om Geometry\label{RSsection}}
Throughout this paper, we employ Eddington-Finkelstein coordinates $(\tau, r, \vartheta, \varphi)$ 
on the domain $\R \times (0,\infty) \times (0,\pi) \times [0,2\pi)$, as introduced in~\cite
{chandra, krpoun-mueller}. In these coordinates, the metric takes the form
\begin{align}
    g & = \frac{\Delta}{r^2}\,\dif{\tau}\otimes \dif{\tau} - \left(2 - \frac
    {\Delta}{r^2}\right)\dif{r}\otimes\dif{r}- \left(1 - \frac{\Delta}{r^2}\right)\left(\dif{\tau}
    \otimes\dif{r} + \dif{r} \otimes \dif{\tau}\right) \notag \\
      & \hspace*{4em}-r^2 \,\dif{\vartheta} \otimes \dif{\vartheta} - r^2\,\sin(\vartheta)^2
      \,\dif{\varphi} \otimes \dif{\varphi} \label{RSmetric}
\end{align}
where $\Delta \equiv \Delta(r) = r^2 - 2Mr + Q^2 = (r-r_+)(r-r_-)$ with $r_\pm = M \pm \sqrt{M^2 - 
Q^2}$ and
\[
    \dif{\tau} = \dif{t} + \dif{u} \:. 
\] 
Here, $u$ denotes the {\em{Regge-Wheeler coordinate}}, also known as the tortoise coordinate, 
defined by
\begin{align}\label{eq:defRW}
    \frac{\dif{u}}{\dif{r}} = \pm \frac{r^2}{\Delta} \:,
\end{align}
where the minus sign corresponds to the incoming and the plus sign to the outgoing null coordinates.
The zeros of the function $\Delta$ define the event horizon and the Cauchy horizon, respectively. 
The region ${r > r_-}$ outside the Cauchy horizon 
constitutes a globally hyperbolic spacetime, whereas the interior region ${r < r_-}$ is not 
globally hyperbolic due to the spacetime singularity at $r = 0$. The topology of the spacetime is 
$\mm \cong \R^2 \times S^2$.

In~\cite{krpoun-mueller}, the Dirac equation was analyzed and separated in a specific gauge, where 
the Dirac matrices are given in the Weyl representation. The Dirac wave function and the Dirac 
operator were then transformed according to
\begin{align}\label{trafo}
    \Psi = \Lambda \psi \quad \text{and} \quad \Gamma \,\Lambda( \Dir - m)\,
    \Lambda^{-1}\Psi = (\rr + \Am) \, \Psi = 0,
\end{align}
with the transformation matrices defined by
\begin{align}\label{DTrafo}
    \Lambda := \dfrac{\sqrt{r}}{r_+}\begin{bmatrix}
        \sqrt{|\Delta|} & 0   & 0   & 0              \\
        0              & r_+ & 0   & 0              \\
        0              & 0   & r_+ & 0              \\
        0              & 0   & 0   & \sqrt{|\Delta|}
    \end{bmatrix}
    \quad \text{and} \quad
    \Gamma := r
    \begin{bmatrix}
        1 & 0  & 0  & 0 \\
        0 & -1 & 0  & 0 \\
        0 & 0  & -1 & 0 \\
        0 & 0  & 0  & 1
    \end{bmatrix},
\end{align}
and the radial and angular operators $\rr$ and $\Am$ as given by the matrices
\begin{align}\label{matricesDirac}
    \rr & := \begin{bmatrix}
                 irm                 & 0                   & |\Delta|^{1/2}
                 \D_0 & 0                   \\
                 0                   & -irm                &
                 0                  & |\Delta|^{-1/2}\D_1 \\
                 |\Delta|^{-1/2}\D_1 & 0                   &
                 -irm               & 0                   \\
                 0                   & |\Delta|^{1/2} \D_0 &
                 0                  & irm                 \\
             \end{bmatrix}\quad \text{and} \\[0.3em]
    \Am & := \begin{bmatrix}
                 0     & 0    & 0     & \L_+ \\
                 0     & 0    & -\L_- & 0    \\
                 0     & \L_+ & 0     & 0    \\
                 -\L_- & 0    & 0     & 0    \\
             \end{bmatrix} .
\end{align}
Here, the linear operators $\Dir_{0/1}, \, \L_{\pm}$ are defined by
\begin{align}
    \D_0     & := - \left(\partial_{\tau} - \partial_r \right)
    \label{operators} \\
    \D_1     & := \left( 2r^2 - \Delta \right)\partial_{\tau} +
    \Delta\partial_r \\
    \L_{\pm} & := \partial_{\vartheta} + \frac{\cot(\vartheta)}{2} \mp i \csc
    (\vartheta)\partial_{\varphi} \label{Lops}
\end{align}
and the spin inner product is chosen to be:
\begin{align}\label{spinProduct}
    \Sl \psi | \phi \Sr_x = - \Big \la \psi,
        \begin{bmatrix}
            0 & \Idmat_{\C^2} \\
            \Idmat_{\C^2} & 0
        \end{bmatrix} \phi \Big \ra_{\C^4} \:.
\end{align}
The form of the spin scalar product can be understood from the requirement that the Dirac matrices 
$\gamma^j$ in Minkowski space be symmetric with respect to this inner product. By employing the 
separation ansatz
\begin{align}\label{sepAnsatz}
    \Psi = e^{-i k \varphi}\sqrt{r_+}
    \begin{pmatrix}
        X_+(\tau, r)\: Y_l(\vartheta)_+        \\
        r_+ \, X_-(\tau, r)\: Y_l(\vartheta)_- \\
        r_+ \, X_-(\tau, r)\: Y_l(\vartheta)_+ \\
        X_+(\tau ,r) \:Y_l(\vartheta)_-
    \end{pmatrix} \quad \text{with} \quad
    \omega \in \R, \, l\, \in \Z, \, k \in (\Z + 1/2)
\end{align}
and using that $\partial_\tau$ is a Killing field ($X(\tau, r) = e^{-i\omega \tau} X(\omega, r)$), 
we obtain eigenvalue problems for both the radial and angular equations. This leads to two 
decoupled ordinary differential equations governing the radial and angular components of the Dirac 
equation:
\begin{align}\label{eq:radialPDE}
    \Bigg(\Delta \begin{bmatrix}
        \partial_r & 0 \\
        0 & \partial_r
    \end{bmatrix} &-
    i \omega\begin{bmatrix}
        \left(2r^2 - \Delta\right) & 0 \\
        0 &-\Delta
    \end{bmatrix} +
    \begin{bmatrix}
        0 & S \\
        \epsilon(\Delta)\bar{S} & 0\\
    \end{bmatrix} \Bigg)
    \begin{pmatrix}
            X_+(\omega,r) \\
            r_+ X_-(\omega,r)
        \end{pmatrix} = 0
\end{align}
and
\begin{align}\label{angularPDE}
    & \left( \begin{bmatrix}
        0     & \L_- \\
        -\L_+ & 0
    \end{bmatrix} - \xi \, \Idmat_{\C^4} \right)
    \begin{pmatrix}
        Y_l(\vartheta)_+ \\
        Y_l(\vartheta)_-
    \end{pmatrix} = 0 \:,
\end{align}
where $\epsilon$ is the sign function
\begin{align}
    \epsilon (x) = \begin{cases}
        1 \quad &\text{if} \quad x > 0 \\
        0 \quad &\text{if} \quad x = 0 \\
        -1 \quad &\text{if} \quad x < 0 \\
    \end{cases}\: ,
\end{align}
$\xi$ is the separation constant, and $S = \sqrt{|\Delta|}(i m r - \xi)$. 
The solutions to the radial equation~\eqref{eq:radialPDE}, together with their asymptotic behavior, 
are discussed in detail in~\cite{krpoun-mueller}. The following lemma summarizes the result for a 
fixed mass value $m$.

\begin{Lemma}\label{le:krpMul}
    For the case $ |\omega| < m$ one solution of the radial
    ODE~\eqref{eq:radialPDE} has exponential decay and the other one exponential
    growth for $u \rightarrow \infty$.
    In the case $|\omega| > m$, on the other hand, the solutions have following
    asymptotics:
    \begin{enumerate}[leftmargin=2em]
        \item[{\rm{(i)}}] {\bf (Asymptotics at infinity)} Let $w_1 \in \C$ be
        the root of $\omega^2 - m^2$ contained in the convex hull of $\R_+ $
        and $\R_+ \cdot i$ and $w_2 = -w_1$
            the other root, and let
            \[ \Theta := \frac{1}{4}\, \ln \Big( \frac{\omega - m}{\omega + m}
            \Big) \:, \]
            then there is~$f^{kl\omega}_m := (f^{kl\omega}_{m\,+}, f^{kl\omega}_{m\,-})^T \in\C^2 
            \setminus \{ 0\} $ with
            \[
                X^\infty_m(\omega,u)
                = \begin{bmatrix}
                    \cosh(\Theta) & \sinh(\Theta) \\
                    \sinh(\Theta) & \cosh(\Theta) \\
                \end{bmatrix}	 \begin{pmatrix}
                    f^{kl\omega}_{m\,+} e^{i\Phi_+(u)}  \\
                    f^{kl\omega}_{m\,-} e^{-i\Phi_-(u)} \\
                \end{pmatrix}
                + E_{\infty}(u)  \]
            \V for the asymptotic phases
            \begin{align}
                \Phi_{\pm}(u):= w_{1}\, u + M \bigg( \pm 2\omega + \dfrac{m^2}
                {w_{1}}\bigg) \ln(u)
                \label{asyPhase}
            \end{align}
            and for an error function $E_{\infty}(u)$ with polynomial decay.
            More precisely, there is $q \in \R_+$ with
            \begin{align*}
                ||E_{\infty}|| \leqslant \dfrac{q}{u}.
            \end{align*}
        \item[{\rm{(ii)}}] {\bf (Asymptotics at the event horizon)} For every non-trivial solution 
        $X$,
            \[ X^H_m(\omega,u) =
                \begin{pmatrix}
                    h^{kl\omega}_{m\,+} e^{2i \omega u} \\
                    h^{kl\omega}_{m\,-}
                \end{pmatrix}
                + E_{+}(u)
            \]
            with~$h^{kl\omega}_m := (h^{kl\omega}_{m\,+}, h^{kl\omega}_{m\,-})^T \in \C^2 \setminus 
            \{0\} $, with $E_{+}$ such that for $r$ sufficiently close to $r_+$ and suitable 
            constants $a,b \in \R_+$,
            \begin{align*}
                ||E_{+}(u)|| \leq a e^{b u}.
            \end{align*}
        \item[{\rm{(iii)}}] {\bf (Asymptotics at the Cauchy horizon)} For every non-trivial 
        solution 
        $X$,
            \[ X^C_m(\omega,u) =
                \begin{pmatrix}
                    c^{kl\omega}_{m\,+} e^{2i \omega u} \\
                    c^{kl\omega}_{m\,-}
                \end{pmatrix}
                + E_{-}(u)
            \]
            with~$c^{kl\omega}_m := (c^{kl\omega}_{m\,+}, c^{kl\omega}_{m\,-})^T \in \C^2 \setminus 
            \{0\} $, with $E_{-}$ such that for $r$ sufficiently close to~$r_-$ and suitable 
            constants $a,b \in \R_+$,
            \begin{align*}
                ||E_{-}(u)|| \leq a e^{- b u}.
            \end{align*}
    \end{enumerate}
\end{Lemma}
Since the above asymptotics at infinity are given by plane waves, their coefficients can be 
interpreted as the transmission and reflection coefficients, analogous to those in classical 
scattering processes.
\begin{Def}\label{def:transmission_vec}
    The asymptotics at infinity $X^\infty_m = (X^\infty_{m\,+}, X^\infty_{m\,-})^T$ describe the 
    outgoing component $X^\infty_{m\,+}$ and the ingoing component $X^\infty_{m\,-}$. We define the 
    corresponding coefficients as the reflection coefficient $f^{kl\omega}_{m\,+}$ and the 
    transmission coefficient $f^{kl\omega}_{m\,-}$. We refer to $f^{kl\omega}_m = (f^{kl\omega}_{m\,
    +}, f^{kl\omega}_{m\,-})^T$ as the transmission vector.
\end{Def}
The angular operator in~\eqref{angularPDE} is essentially self-adjoint on $L^2(S^2, \C^2)$ with 
dense domain $C^\infty(S^2, \C^2)$ and has a purely discrete spectrum (for details, see~\cite
[Section~3]{tkerr}). More precisely, the angular operator coincides with the spin-weighted 
spherical operator for $s = \tfrac{1}{2}$, as analyzed in~\cite{goldberg}. We denote the 
corresponding orthonormal eigenvector basis by $_{\frac{1}{2}}Y_{kl} = e^{-i k \varphi} \,Y_{l}
(\vartheta)$ with $l \in \Z$ and $k \in \Z + \frac{1}{2}$, i.e.\
\begin{align}\label{angularON}
    \braket{e^{-i k\varphi}\, Y_{l}(\vartheta)}{e^{-i k'\varphi}\, Y_{l'}(\vartheta)}_{L^2(S^2)}
    = \delta_{k,k'}\:\delta_{l,l'} \:.
\end{align}
Therefore, it has a spectral decomposition of the form
\begin{align}\label{eq:specdecomp_angular}
    \mathcal{A} = \sum_{k,l} \xi_l \, A_{k, l}  \, ,
\end{align}
where $A$ is the spectral measure of the operator $\mathcal{A}$.

In~\cite{finster-krpoun1}, the integral representation of the Dirac propagator in the Reissner--Nord\-ström spacetime in Eddington-Finkelstein coordinates was derived. This derivation 
relies on the fact that the Hamiltonian $H$ is essentially self-adjoint with respect to the 
conserved scalar product on its domain of definition.
\[
    C^\infty_\text{init}(N) := \big\{ \psi_0 \in C^{\infty}_0(N \, , S M) \text
    { with }
    (\slashed{n} - i)\,(H^p\,\psi_0)\big|_{\partial N} = 0 \text{ for all~$p
    \in \N$} \big\} .
\]
Here, $N = N_\tau \big|_{\tau=0}$ denotes the initial spacelike hypersurface arising from the 
spacelike foliation $(N\tau)_{\tau \in \R}$ of $\mm$ associated with the Cauchy problem, whose 
boundary satisfies $\partial N \simeq S^2$. Moreover, $n$ denotes the inward-pointing normal on 
$\partial \mm$ (see~\cite{chernoff, finster-krpoun1} for further details).
\begin{Def}\label{conservedSP}
    The conserved scalar product on solutions of the Dirac equation has the
    form
    \begin{align}
        \Scpr{\psi}{\phi} = \int_{r_0}^{\infty} \dif{r} \int_{-1}^{1}\dif{\cos(\vartheta)} 
        \int_0^{2\pi} \dif{\varphi}\;\Psi^{\dagger}(\tau, r, \vartheta, \varphi) \,\Sigma(r)  \,\Phi
        (\tau, r, \vartheta, \varphi) ,
    \end{align}
    where the matrix $\Sigma$ is described by
    \[
        \Sigma(r)  = \dfrac{r_+}{|\Delta|}
        \begin{bmatrix}
            2r^2 - \Delta & 0        & 0        & 0             \\
            0             & |\Delta| & 0        & 0             \\
            0             & 0        & |\Delta| & 0             \\
            0             & 0        & 0        & 2r^2 - \Delta \\
        \end{bmatrix} .
    \]
\end{Def}
Using the spectral theorem for unbounded self-adjoint operators,
\begin{align}\label{genIntRep}
    \psi(\tau) = e^{-i\tau H} \psi_0 = \int_{\sigma(H)} e^{-i\omega \tau} \: \dif{E}{\omega}\, 
    \psi_0 ,
\end{align}
The main result of~\cite[Theorem~5.5]{finster-krpoun1} is an expression of the spectral measure in 
terms of two globally defined fundamental solutions of the radial ODE~\ref{eq:radialPDE}. Since 
fundamental solutions are not unique, they admit a $\GL{2}{\C}$ freedom, which can be used to 
relate different choices of fundamental solutions to one another. All physical quantities, such as 
the resolvent or the Green's matrix, remain invariant under these transformations.

A distinguished class of such solutions is given by the so-called Jost solutions, denoted by $\jc{}
$ and $\jh{}$, which were introduced in~\cite[Lemma~4.3]{finster-krpoun1}. They are derived from an 
asymptotic expansion of the radial ODE in Schrödinger-type form, complexified in the frequency 
parameter $\omega \rightarrow \omega \pm i\varepsilon$. These solutions are used to construct the 
above-mentioned global fundamental solutions (see~\cite[Lemma~4.4]{finster-krpoun1}), which are 
physical $L^2$-functions on the entire region $(r_-, \infty)$. The procedure for obtaining $L^2_
{\mathrm{loc}}$-extensions is described in detail in~\cite[Section~4.3]{finster-krpoun1} and relies 
on matching conditions at the horizons together with the convergent behavior of the complexified 
Jost-type solutions.

To this end, we define the Wronskian matrix which we also refer to as the fermionic flux matrix (the 
notation is motivated by the gamma matrices appearing in the current vector), which is derived 
in~\cite[Lemma~4.6]{finster-krpoun1}.
\begin{Def}\label{def:wronskian}
    On the real axis the Wronskian is defined by
    \begin{align}
        W(\Phi)(r) := W(\Phi_i, \Phi_j)(r) = \la \Phi_i(r), A^r\, \Phi_j(r) \ra_{\C^2} \quad \text
        {with} \quad i,\, j \in \{1,2\} \, ,
    \end{align}
    with
    \[
        A^r = 
            \begin{bmatrix}
                1 & 0 \\
                0 & - \epsilon(\Delta)\\
            \end{bmatrix} \, .
    \]
    It is independent of $r$ within the respective regions $(r_-, r_+)$ and $(r_+, r_\infty)$.
\end{Def}
For the result of Theorem~5.5 in~\cite{finster-krpoun1}, we used the freedom in the fundamental 
solutions and chose a specific basis in order to simplify the computations. We now describe a 
procedure for choosing the initial conditions that determine the set of global fundamental 
solutions $\{\Phi_1, \Phi_2\}$ uniquely. The boundary conditions for the solution~$\Phi_1$ are set via its 
asymptotics at infinity, given by 
\[
    \lim_{r\nearrow \infty}\Phi_1(r) = \lim_{r\nearrow \infty} \jh{\infty}(\OmegaVar, r)\, ,
\]
while $\Phi_2$ satisfies initial conditions from a linear combination of the asymptotic solutions 
at the Cauchy horizon of the form
\begin{align}\label{eq:bc_phi2}
    \lim_{r\searrow r_-} \Phi_2(r) =  \lim_{r\searrow r_-}\Big(\jc{-}(\OmegaVar, r) + c(\OmegaVar) 
    \jh{+}(\OmegaVar, r) \Big) \,.
\end{align}   
This can be derived by taking the difference of the general fundamental solutions and evaluating 
them at the Cauchy horizon where $r=r_-$. In the end, one performs
the limit~$\varepsilon \searrow 0$. 
The $\omega$-dependence of $c(\omega)$ is determined by the matching conditions arising in the 
construction of the global fundamental solutions.

In addition, we express $\Phi_2$ in terms of the Jost solutions at infinity, which are then 
extended to the entire interval $(r_-, \infty)$:
\[
    \Phi_2 = a \,\jh{\infty} + b\, \jc{\infty} \, .
\]
Their limits on the real line are given by
\begin{align}\label{eq:constr_fund}
    \chi_{1,m}(\omega, r) = \lim_{\varepsilon \searrow 0} \jh{\infty}(\OmegaVar, r) 
    \quad \text{and} \quad 
    \chi_{2,m}(\omega, r) = \lim_{\varepsilon \searrow 0} \jc{\infty}(\OmegaVar, r) \, .
\end{align}
These $\chi_{i,m}$ are used to compute the integral representation. More precisely, they arise as 
the combination of the radial solution and the spin-weighted spherical harmonics, as we will 
explain later. We want to highlight that, since $\Phi_2$ is constructed at the Cauchy horizon, the 
resulting coefficients $a$ and $b$ are determined by the global behavior of the fundamental 
solutions up to the Cauchy horizon, which is an essential difference from~\cite{sigbh}.

The coefficients $a$ and $b$ are pseudo-normalized, as can be seen by evaluating the Wronskian (see 
Definition~\ref{def:wronskian}) with signature $(1,-1)$ at infinity and computing the $(2,2)
$-component of the corresponding Gram matrix. Noticing that $\chi_{i,m}$ are defined by the 
normalized Jost functions at infinity, the Gram matrix of the Wronskian is constant $\diag(1,-1)$ 
on $(r_+, \infty)$ for any linear combination. Consequently, they can be parametrized in the 
general form
\begin{align}\label{eq:def_ab}
    a = e^{i \alpha} \sinh(\vartheta) \quad \text{and} \quad b = e^{i \beta} \cosh(\vartheta) \, ,
\end{align}
where $\alpha, \beta, \vartheta \in \R$. At this point, we emphasize that $\vartheta$ is a function 
of the energy parameter $\omega$ as well as of the angular modes $(l,k)$. Adhering to the 
conventions introduced above, the parameters $a$ and $b$ represent the reflection and transmission 
coefficients of the fundamental solution $\Phi_2$. Using Definition~\ref{def:transmission_vec}, we 
combine them into the transmission vector
\begin{align}\label{eq:def_transvec}
    f^{kl\omega}_{2, m} =
        \begin{pmatrix}
            e^{i \alpha} \sinh(\vartheta) \\
            e^{i \beta} \cosh(\vartheta)
        \end{pmatrix} \, ,
\end{align}
which has $\C^2$-norm $\|f^{kl\omega}_{2, m}\|_{\C^2} = \cosh(2 \vartheta)$. Consequently, we can 
express $\vartheta$ as a function of the transmission vector $f^{kl\omega}_{2, m}$,
\begin{align}\label{eq:vartheta}
    \vartheta(\omega, l, k) = \dfrac{1}{2} {\rm{arccosh}}(\|f^{kl\omega}_{2,m}\|_{\C^2}) \, .
\end{align}
We summarize the results of this section in the following way. A general fundamental solution of 
the Dirac has the form
\[
    \Psi_{i,m}^{kl}(\omega, r,\, \vartheta,\varphi) := \chi_{i,m}(\omega, r)_{\frac{1}{2}}Y^{kl}
    (\vartheta,\varphi) \,,
\]
with $k \in \big(\Z + \frac{1}{2}\big), \: l\in \Z \text{ and } p \in \{1,2\}$ and whereas $_{\frac
{1}{2}}Y^{kl}(\vartheta,\varphi)$ are the orthonormal basis vectors of the spin$-1/2$ weighted 
spherical harmonics~\eqref{angularON} and $\chi_{p,m}(\omega,r)$ denotes the constructed radial 
fundamental solutions~\eqref{eq:constr_fund} which have, via construction, transmission vectors of 
the from 
\[
    f^{kl\omega}_{m, 1} = \begin{pmatrix}
        1 \\
        0
    \end{pmatrix} \quad \text{and} \quad 
    f^{kl\omega}_{m, 2} = \begin{pmatrix}
        0 \\
        1
    \end{pmatrix} \, .
\] 
Additionally, for $|\omega| < m$, at asymptotic infinity we set the exponential decaying solution 
and its  extension for $\chi_{m,1}$, whereas the solution for $\chi_{m,2}$ increases exponentially.

Next, we state the main result from~\cite{finster-krpoun1} for a general fundamental solution $\Psi_
{p,m}^{kl}$ and express the coefficient matrix $T$ in terms of the components of the 
transmission vector~\eqref{eq:def_transvec} defined above. Also, we added a mass index to indicate 
that this is a representation for a fixed mass $m$.
\begin{Lemma}\label{le:integralRep}
    The Dirac propagator in the Reissner-Nordström geometry with \\
    Eddington-Finkelstein coordinates can be expressed in terms of the general fundamental solutions 
    $\Psi^{kl}_{i,m}(\omega, r, \vartheta, \varphi)$ with $i \in \{1,2\}$ in $r\, \in (r_-, \infty)$ as
    \begin{align}
        \psi_m(\tau, r, \vartheta, \varphi) = \sum_{k,l} \int_{\R \setminus \{\pm m\} } e^{-i\omega \tau} \, \sum_{i=1}
        ^2 \,\widehat{\psi}^{kl}_{i,m}(\omega) \,\Psi^{kl}_{i,m}(\omega,r,\vartheta,\varphi)\:\dif
        {\omega}
    \end{align}
    with $\widehat{\psi}^{kl}_{i,m}(\omega) : \C \rightarrow \C$ smooth functions defined by
    \[
        \widehat{\psi}^{kl}_{i,m}(\omega) = \dfrac{1}{2\pi}\sum_{j=1}^2 t_{ij}(\omega, l)\, (\Psi^{kl}_{j,m}
        (\omega)\,|\, \psi_{m,0})\:.
    \]
    Here~$(\cdot| \cdot )$ denotes the conserved scalar product on the hypersurfaces defined in 
    definition~\ref{conservedSP}, restricted to the upper left $2 \times 2$ block. Moreover, $t_{ij}
    $ are the entries of the matrix $T$ and~$\psi_{m,0}(r) \, \in C^\infty_\text{\rm{init}}(N)$.
    Depending on $\omega$ and $m$, the matrix~$T$ has the entries
    \begin{align*}
        T = \dfrac{1}{2} 
        \begin{cases}
            \begin{pmatrix}
                q & 0 \\
                0 & 0 
            \end{pmatrix} \text{ with } q \in \R  &\text{ if~$|\omega|<m$} \\[2em]
            \begin{pmatrix} 
                1 & t \\
                t & 1 
            \end{pmatrix},  &\text{ if~$|\omega|>m$}\:. 
        \end{cases}
    \end{align*}
    with $ t = e^{-i(\alpha - \beta)} \tanh(\vartheta)$, where $\alpha, \, \beta$ and $\vartheta$ real 
    functions of $\omega, \, k$ and $l$.
\end{Lemma}
\begin{proof}
    This follows directly from~\cite[Theorem~5.5]{finster-krpoun1} and the identity relation for 
    the spin-weighted spherical harmonics
    \[
        \sum_{k,l} {}_{\frac{1}{2} }Y_{kl} \otimes {}_{\frac{1}{2}}Y_{kl}^*= \Idmat_{L^2(S^2)} \, ,
    \]
    with $Y^*$ the adjoint with respect to the $L^2(S^2)$ product. The sum converges because it is 
    a decomposition into the orthogonal eigenspace of $\mathcal{A}$.
    
    The entries of $T$ follow from a short computation using the expressions for $a$ and $b$ 
    from the above discussion defined in equation~\eqref{eq:def_ab}. Additionally, we absorbed one 
    factor $1/2$ into the matrix $T$, whereas the factor $1/r_+$ is canceled by the prefactors 
    arising from the ansatz for the radial solution in~\ref{sepAnsatz}, upon imposing the 
    transformation $\sqrt{r_+}X(\tau,r) \rightarrow X(\tau,r)$.
\end{proof}

Additionally, we choose a dense subset $\hh^{\infty} \subset \hh$, spanned by families of solutions 
to the Dirac equation. This choice ensures suitable analytical properties for handling technical 
aspects such as differentiability and integrability.
\begin{Def}
    \label{subsetH}
    The domain $\hh^{\infty} \subset \hh$ is chosen as the space of all Dirac solutions of the form
    \eqref{le:integralRep} which satisfy the following conditions:\vspace*{0.5em}
    \begin{enumerate}
        \item The functions $\widehat{\psi}_{i,m}^{kl}(\omega)$ vanish identically for almost all 
            $k \in (\Z + 1/2)$ and $l \in \Z$. \\
        \item For all $k \in (\Z + 1/2)$, $l \in \Z$ and $i \in \{1,2\}$, the functions $\widehat
            {\psi}_{i, m}^{kl}(\omega)$ are smooth and compactly supported in $\omega$ and $m$. 
            Moreover, they are supported away from $\omega = \pm m$, i.e.
            \begin{align}
                \supp{\widehat{\psi}_{i, \cdot}^{kl}(\cdot)} \subset \big\{(\omega, \, m) \, \in \, 
                \R \times I \quad \text{with} \quad \omega \neq \pm m\big\}
            \end{align}
    \end{enumerate}
\end{Def}
This definition leads to two main results. First, it ensures that the Dirac solutions in 
$\hh^\infty_m$ are well defined. Second, it excludes solutions that asymptotically exhibit 
vanishing momentum at infinity. This exclusion is justified by the fact that $\hh^\infty$ is a 
dense subset. The denseness follows from the absence of a point spectrum in the Dirac Hamiltonian, 
as shown in~\cite{finster-krpoun1}, and is further supported by the integral representation of the 
Hamiltonian.

\section{The Mass Decomposition Theorem in Reissner-Nordström}\label{sec:mass_decomp}
\subsection{Preparations}
In the context of the Reissner-Nordström geometry, the results of~\cite{finite, infinite} do not 
apply. The main reason is that waves can cross the event horizon of the black hole and disappear from the observable region.
On a more technical level, the mass oscillations are not present
on the event horizon due to the infinite red shift effect.
Another complication in our setting is that
two horizons are present.

Nevertheless, having derived the integral representation of the Dirac propagator, this 
representation can be used to formulate a so-called mass decomposition. This approach was studied 
in detail in~\cite[Section~3]{sigbh} for the exterior Schwarzschild geometry. The following steps 
extend that result to the Reissner-Nordström geometry up to the Cauchy horizon.

When evaluating the inner product, we can exploit the orthogonality of the angular eigenfunctions 
to simplify the expression. Using~\eqref{indefitSP}, \eqref{spinProduct}, and~\eqref
{le:integralRep}, we obtain for all $\psi, \phi \in \hh^\infty$
\begin{align}\label{eq:simplInnerProduct}
    \braket{\PP \psi}{\PP \phi} &= -2 \int_{-\infty}^\infty \dif{t}
    \int_{r_-}^\infty \dif{r} \dfrac{r}{\sqrt{|\Delta|}} \int_I \dif{m} \int_I
    \dif{m'} \int_{-\infty}^\infty \dif{\omega} \int_{-\infty}^\infty \dif
    {\omega'} \, e^{i(\omega\, -\, \omega')t}\notag \\
    & \hspace{6em} \times \int_{0}^{2\pi}\dif{\varphi} \int_{-1}^{1} \dif
    {\cos(\vartheta)} \sum_{k,l,k',l'} \sum_{i, j} \overline{\hat{\psi}^{kl}(\omega)}
    _{i, m} \, \hat{\phi}^{kl}(\omega')_{j, m'} \notag \\
    & \hspace{10em}\times \braket{_{\frac{1}{2}}Y^{kl} \chi_{i,m}(\omega, r)
    (\vartheta,\varphi)}{A^0\, _{\frac{1}{2}}Y^{k'l'} \chi_{j,m'}(\omega', r)
    (\vartheta,\varphi)}_{\C^2} \notag \\
    & = - 2 \sum_{k,l} \int_{r_-}^\infty \dif{r} \dfrac{r}{\sqrt{|\Delta|}}
    \int_I \dif{m} \int_I \dif{m'} \int_{-\infty}^\infty \dif{\omega}
    \notag \\
    & \hspace{5em} \times \sum_{i, j} \overline{\hat{\psi}^{kl}(\omega)}_{i, m} \, \hat{\phi}^{kl}
    (\omega)_{j, m'} \braket{\chi_{i, m}(\omega, r)}{A^0\chi_{j, m'}(\omega, r)}_{\C^2}
\end{align}
where we get a factor 2 by expressing the four-spinor $X$ as two two-spinors and defined the matrix
\[
    A^0 := \begin{bmatrix}
        0 & 1 \\
        1 & 0 \\
    \end{bmatrix}\, ,
\] for a short notation. Additionally, we can express the scalar product in an alternative way
\begin{Lemma}
    \label{scalLemma}
    For any Dirac solutions $\phi_m, \psi_m \in \hh^\infty$, the scalar product~\eqref{scalPro} can 
    be written in the alternative form
    \begin{align}
        \Scpr{\psi_m}{\phi_m}_m = 2\pi \sum_{k,l} \int_{\R} \sum_{i,j = 1}^2 (T^{-1})^{ij} \, 
        \overline{\widehat{\psi}^{kl}_{i,m}(\omega)} \, \widehat{\phi}^{kl}_{j,m}(\omega) \, \dif
        {\omega}\: ,
    \end{align}
where $T$ is the matrix from Lemma~\ref{le:integralRep}.
\end{Lemma}
\begin{proof}
    By looking at the integral representation for $t=0$ we end up with
    \begin{align*}
        \psi_m\big|_{t=0} = \dfrac{1}{2\pi} \sum_{k,l} \int_{\R} \, \sum_{i,j=1}^2 \, t_{ij} \, 
        (\Psi^{kl}(\omega)_{j, m}\,|\, \psi_m)\big|_{\tau = 0} \,\Psi^{kl}(\omega)_{i, m}(r,\, 
        \vartheta,\, \varphi)\:\dif{\omega}
    \end{align*}
    Since the scalar product is conserved in time, the scalar product of two solutions at any 
    arbitrary time $\tau$ equals that at $\tau = 0$, i.e.,
    $\Scpr{\psi_m}{\phi_m}_m = \Scpr{\psi_m}{\phi_m}_m\big|_{\tau=0}$. Taking the scalar product 
    with another solution $\phi_m$ and applying Lemma~\ref{le:integralRep} yields the result.
\end{proof}
When looking at the scalar product~\ref{eq:simplInnerProduct}, we observe that only pairings 
between solutions with the same angular momentum contribute. This is also the case for the inner 
product from the above Lemma~\ref{scalLemma}. Therefore, we restrict our subsequent computations to 
a single wave mode $(k, l)$. For readability, we consistently omit the corresponding sums and 
indices. We now proceed with a short lemma.
\begin{Lemma}
    \label{decay}
    For all $\psi \in \hh^\infty$ and $\forall \, \omega \in \, \R$ we have
    following radial decay
    \begin{align}
        \int_I \dif{m} \, \sum_{i=1}^2 \, \widehat{\psi}_{i,m}(\omega) \chi_{i,m}(\omega, r) = 
        \mathcal{O}\bigg(\frac{1}{r}\bigg) .
    \end{align}
\end{Lemma}
\begin{proof}
    By Definition~\ref{subsetH}, the functions $\widehat{\psi}_{i,m}(\omega)$ and, from Lemma~\ref
    {le:krpMul}, the asymptotic phases are smooth with respect to $m$. In the case $|\omega| < m$, 
    only the decaying solution $\chi_{1,m}$ remains because $T = {\rm{diag}(f,0)}$ and the 
    growing solution is suppressed, ensuring the required behavior for large $r$. 
    
    For $|\omega| > m$, $\chi_{i,m}$ corresponds to the Jost solutions at asymptotic infinity. As 
    one can verify in~\cite[Lemma~4.3]{finster-krpoun1}, the error term can be neglected because it 
    decays polynomially in $r$. Since the remaining term is smooth, the Riemann--Lebesgue lemma 
    applies, yielding a rapidly decaying contribution as $r \to \infty$.
\end{proof}
Next, we examine the behavior of the mass multiplication operator, which will later be used to 
determine the regular part of the mass integrals.
\begin{Def}\label{massMulOp}
    We define the mass multiplication operator $\mathcal{T}_m: \hh^\infty \rightarrow
        \hh^\infty$ as
    \[
        \mathcal{T}_m\,\psi_m:= m \, \psi_m .
    \]
\end{Def}
We now proceed to compute the boundary term, which breaks the symmetry of the multiplication 
operator with respect to the Lorentz-invariant inner product~\ref{indefitSP}. In contrast to the 
exterior case (see~\cite{sigbh}), we obtain additional contributions from the event and Cauchy 
horizons.
\begin{Lemma}
    \label{le:non_singular}
    For all $\psi, \phi \in \mathcal{H}^\infty$, we have that $\mathcal{T}_m$ is not symmetric with 
    respect to the Lorentz invariant inner product~\ref{indefitSP} and differs by
    \begin{align}
        \braket{\PP\psi}{\PP \mathcal{T}_m \phi} - \braket{\PP \mathcal{T}_m \psi}{\PP \phi} = 2
        i  \int_I \dif{m} \, &\int_I \dif{m'}
        \int_{-\infty}^\infty \dif{\omega} \notag \\
        & \times \sum_{i,j = 1}^2 \overline{\widehat{\psi}_{i,m}(\omega)} \, \widehat{\phi}_{j,m'}(\omega)\, \Omega(\omega,
        i, j)_{ij}\: .
    \end{align}
    Here, $\Omega(\omega,m,m') \in {\rm{Mat}}(2,\C)$ describes a complex matrix (boundary matrix) 
    with entries
    \[
        \Omega(\omega,m,m')_{ij} = W(\chi^\omega_{i,m},\chi^\omega_{j,m'})(r^>_+) - W(\chi^\omega_
        {i,m},\chi^\omega_{j,m'})(r^<_+) - W(\chi^\omega_{i,m},\chi^\omega_{j,m'})(r_-)
    \] 
    where $r^>_+$ indicates that the Wronskians lies inside $(r_+, \infty)$ and $r^<_+$ inside $
    (r_-, r_+)$.
\end{Lemma}
We emphasize that $\Omega$ depends on the wave modes $(k,l)$, thus the full expression 
reads $\Omega^{kl}(\omega, m , m')$.
\begin{proof}[Proof of Lemma~\ref{le:non_singular}.]
    Since the angular momentum modes are fixed, we can focus on the radial
    ODE. We start by re-writing~\eqref{eq:radialPDE} as
    \[
        \mathcal{D} \, \psi = m\, \psi \:,
    \]
    with $\mathcal{D}$ the Dirac operator expressed by
    \[
        \mathcal{D} = i \dfrac{\sqrt{|\Delta|}}{r}
        \begin{bmatrix}
            0 & -\epsilon(\Delta) \\
            1 & 0  \\
        \end{bmatrix} \partial_r
        + \dfrac{\omega}{r \sqrt{|\Delta|}}
        \begin{bmatrix}
            0 & |\Delta| \\
            2r^2 - \Delta & 0 \\
        \end{bmatrix} +
        \frac{i \, \epsilon(\Delta)}{r}
        \begin{bmatrix}
            \xi & 0  \\
            0 & -\xi \\
        \end{bmatrix}.
    \]
    By direct computation, one verifies that the second and third terms are symmetric with respect 
    to the spin inner product~\eqref{spinProduct}. In the next step, we evaluate the difference in 
    the inner product when applying the mass oscillation operators~\eqref{eq:simplInnerProduct}, 
    with the mass multiplication operator acting on the left and right sides.
    \begin{align}\label{diffProducts}
        \braket{\PP\psi}{\PP \mathcal{T}_m \phi} &- \braket{\PP \mathcal{T}_m\psi}{\PP \phi} = -2
        i \int_{r_-}^\infty \dif{r} \int_I \dif{m} \, \int_I \dif{m'}
        \int_{-\infty}^\infty \dif{\omega} \\ \nonumber
        &\times  \sum_{i, j}\overline{\widehat{\psi} _{i,m}(\omega)} \, \widehat{\phi}_{j,m'}
        (\omega) \Bigg( \Bigg\la
            \begin{bmatrix}
                0 & -\epsilon(\Delta) \\
                1 & 0  \\
            \end{bmatrix} \partial_r  \chi_{i, m}(\omega,r) \Bigg| A^0 \chi_{j, m'}
            (\omega,r) \Bigg\ra_{\C^2} \\ \nonumber
            & \hspace{4em} + \braket{\chi_{i, m}(\omega,r)}{A\, \partial_r \chi_{j, m'}
            (\omega,r)}_{\C^2} \Bigg)
    \end{align}
    We can apply Fubini's theorem to interchange the integrals, using the properties of $\hh^\infty$
    defined in~\ref{subsetH} to evaluate the above expression via partial integration.

    We observe that the matrix $A^r$ undergoes a sign change in its $(2,2)$-component, leading to a 
    finite discontinuity when crossing $r_+$ on a measure zero set. To deal with the sign flip, we 
    split the integration into two regions, $(r_-, r_+)$ and $(r_+, \infty)$, on each of which the 
    partial integration is well defined, and we then analyze their sum. A short computation yields
    \begin{align}
        \braket{\PP\psi}{\PP \mathcal{T}_m \phi} &- \braket{\PP \mathcal{T}_m \psi}{\PP \phi} = - 2
        i \int_I \dif{m} \, \int_I \dif{m'} \, \int_{-\infty}^\infty \dif
        {\omega} \notag \\
        & \times \, \sum_{i, j} \overline{\widehat{\psi}_{i, m}(\omega)}
        \widehat{\phi}_{j, m'}(\omega)\Bigg\{\braket{\chi_{i, m} (r, \omega)}{A^r\, \chi_{j, m'}(r, 
        \omega)}_{\C^2}\bigg|_{r_+}^\infty \notag \\
        &\hspace*{6em}+ \braket{\chi_{i, m} (r, \omega)}{A^r\, \chi_{j, m'}(r, \omega)}_{\C^2}
        \bigg|_{r_-}^{r_+}\Bigg\}
    \end{align}
    According to Lemma~\ref{decay}, the boundary term at spatial infinity vanishes, leaving only 
    the boundary contributions from the Cauchy and event horizons. Using Definition~\ref
    {def:wronskian} gives $\Omega(\omega, m,m')$ as a linear combination of Wronskians.
\end{proof}
Next, we compute the singular contribution to the spacetime inner product for $m = m'$. Since 
the computation is identical that in~\cite[Lemma~3.7]{sigbh}, we merely state the result and refer to~\cite{sigbh} for the proof 
(in our case, the $f^{kl\omega}_{i,m}$ are already orthogonal, which simplifies the last step of the 
proof).
\begin{Lemma}\label{singular}
    For all $\psi, \, \phi \in \hh^{\infty}$ we have
    \begin{align}
    \braket{\PP \psi}{ \PP \phi} &= 2\pi \int_I \dif{m} \int_{\R \backslash
    [-m,m] } \epsilon(\omega)\,
    \dif{\omega} \sum_{i, j = 1}^2 \overline{\widehat{\psi}_{i, m}(\omega)} \widehat{\phi}_{j, m}
    (\omega)\braket{f^\omega_{i, m}}{ f^\omega_{j, m}}_{\C^2} \\
    & -2 i \int_I \dif{m} \int_I \dif{m'} \,\pv{\frac{1}{m-m'}}  \notag \\
    &\hspace{4em}\times \int_{\R \backslash [-m,m] }\Big(\overline{\widehat{\psi}_{1,m} 
    (\omega)} \widehat{\phi}_{1,m'}(\omega) - \overline{\widehat{\psi}_{2,m} 
    (\omega)} \widehat{\phi}_{2, m'}(\omega)\Big)\dif{\omega} \\
    &+ \int_I \dif{m} \int_I\dif{m'} \, E(m, \, m') \notag
    \end{align}
        with $E \in L^\infty(I \times I)$ and $\rm{PV}$ being the Cauchy
        principal value.
\end{Lemma}
\subsection{The Mass Decomposition Theorem}
We can now combine the previous results in the following theorem, which describes the mass 
decomposition of the spacetime inner product. This decomposition allows one to identify the 
signature operator by comparing the spacetime inner product with the conserved scalar product on 
Cauchy hypersurfaces.
\begin{Thm}\label{thm:massDecompTheo}
    For all $\psi, \, \phi \, \in \hh^\infty$ we have the identity
    \begin{align}
        \braket{\PP\psi}{\PP\phi} & = 2 \pi \int_I \dif{m} \int_{|\omega|>m} \epsilon(\omega)\,\dif
        {\omega} \, \sum_{i,j = 1}^2 \overline{\widehat{\psi}_{i, m}(\omega)}\,
        \widehat{\phi}_{j, m}(\omega)\braket{f^{\omega}_{i,m}}{f^{\omega}_{j,m}}_{\C^2}\label
        {massDecompSing} \\
        &-2 i \int_I \dif{m} \, \int_I \dif{m'}\, \pv{\dfrac{1}{m-m'}}\notag \\
        &\hspace{6em} \times \int_{\R}\dif{\omega} \sum_{i, j=1}^2 \overline{\widehat{\psi}_{i, m}(\omega)}\, \widehat{\phi}_
        {j, m'}(\omega)\, \Omega(\omega,m,m')_{ij}\label{massDecompReg}\: ,
    \end{align}
    where $\overline{\widehat{\phi}_{i,m}(\omega)}$ and $\widehat{\psi}_{j, m}(\omega)$ are the 
    smooth coefficients from Lemma~\ref{le:integralRep}, $f^{\omega}_{i, m} \in \C^2$ are the 
    transmission vectors defined in~\ref{def:transmission_vec} and $\Omega(\omega, m, m')$ the 
    boundary matrix from Lemma~\ref{le:non_singular}.
\end{Thm}
\begin{proof} $\;$
    We derive the singular behavior using the results from Lem\-ma~\ref{le:non_singular} and
    Lemma~\ref{singular}. We emphasize that the integrals can be interchanged due to the properties 
    of $\hh^\infty$ defined in Definition~\ref{subsetH}. Our goal is to obtain an explicit 
    expression for the integral kernel $A(m,m')$ from
    \begin{align}
        \braket{\PP \psi}{\PP \phi} = \int_I \dif{m} \, \int_I \dif{m'} \, A(m,m')\, .
    \end{align}
    From Lemma~\ref{singular}, we know that $A$ is a distribution with a pole singularity on the 
    mass diagonal. Hence, it has the general form
    \[
        A(m,m') = \pv{\dfrac{1}{m-m'}}B(m,m') + i \pi \delta(m-m') C(m) + E(m,m') \, ,
    \]
    with $B(m,m')$ and $C(m)$ some regular functions. In the following, we determine these 
    functions with the help of Lemmas~\ref{le:non_singular} and~\ref{singular}. We begin with the 
    case $m \neq m'$ and assume that $\widehat{\psi}_{i,m}(\omega)$ and $\widehat{\phi}_{j,m'}
    (\omega)$ have disjoint mass support. In this situation, we can apply the results from 
    Lemma~\ref{le:non_singular} and compare them with the relation
    \begin{align}
        \braket{\PP \, \psi}{\PP \, T \, \phi} & - \braket{\PP \, T\, \psi}{\PP 
        \, \phi} = - 2 \int_{r_-}^\infty \dif{r}\dfrac{r}{|\Delta|^
        {1/2}} \int_I \dif{m} \, \int_I \dif{m'} \, \int_{-\infty}^\infty \dif
        {\omega} \notag \\ 
        & \times (m'-m)\sum_{i, j} \overline{\widehat{\psi}_{i, m}(\omega)}\widehat{\phi}_{j, m'}(\omega) 
        \braket{\chi_{i, m} (r, \omega)}{A^0 \chi_{j, m'} (r, \, \omega)}_{\C^2}\:,
    \end{align}
    where we used the definition of the mass multiplication operator~\ref{massMulOp}. This leads to 
    an integral kernel of the form
    \begin{align}
        \label{kernel}
        A(m,m') =-\dfrac{2 i}{m - m'} \int_{-\infty}^\infty \sum_{i,j= 1} 
        ^2 \overline{\widehat{\psi}_{i, m}(\omega)} \widehat{\phi}_{j, m'} (\omega)\, \Omega(
            \omega,m,m')_{ij} \dif{\omega} \, .
    \end{align}
    By generalizing to all mass supports and using the distributional relation 
    (sometimes referred to as the Sokhotski--Plemelj formula)
    \[
        \lim _{\varepsilon \searrow 0 }\dfrac{1}{m'-m +i\varepsilon} = \pv{\dfrac{1}{m' - m}} - i 
        \pi \delta(m' - m) \, ,
    \]
    we reproduce the singular terms from Lemma~\ref{singular}. This can be seen directly by 
    comparing the above expression with Lemma~\ref{singular}. Off the diagonal, the comparison with 
    Lemma~\ref{singular} shows that the principal value is already completely determined. Since the 
    singularity is a simple pole, the Sokhotski--Plemelj formula ensures that the corresponding 
    distributional expansion is unique. Consequently, the contribution on the mass diagonal 
    must also coincide with Lemma~\ref{singular}.

    In the final step, we need to show that for the case $|\omega| < m$, the integral involving the 
    principal value can be absorbed into the bounded function $E(m,m') \in L^\infty(I \times I)$. 
    To this end, we analyze the behavior of the matrix components and perform a Taylor expansion in 
    the mass difference $(m-m')$:
    \[
        \Omega(\omega,m,m')_{11} = \Omega(\omega,m,m)_{11} + \mathcal{O}(m-m')
    \]
    (all other components vanish due to $T$ for $|\omega| < m$ as defined in Lemma~\ref
    {le:integralRep}). Since $\Omega$ is, by Lemma~\ref{le:non_singular}, a linear combination of 
    Wronskians, we can use the $r$-independence for $m=m'$. Because $W(\chi)(r^<_+)$ and $W(\chi)
    (r_-)$ lie in the same regularity region, they cancel each other due to the opposite sign. We 
    are left with
    \[
        \Omega(\omega,m,m')_{11} = W(\chi^\omega_{1,m}, \chi^\omega_{1,m})(r_-) + \mathcal{O}(m-m') 
        \, .
    \]
    With the same argument as before, we can relate the Wronskian to its value at spatial infinity. 
    Since $\chi_{1,m}$ describes the exponentially decaying solution for $|\omega| < m$ and the 
    Wronskian is constant, the first term vanishes. Thus, we conclude that
    \begin{align}\label{eq:zeroWronskian}
        \Omega(\omega,m,m')_{11} = \mathcal{O}(m-m') \, ,
    \end{align}
    which yields the result.
\end{proof}
\begin{Remark} (Interpretation of the boundary term)\label{rmk:flux}
{\rm{We emphasize that the second term~\eqref{massDecompReg} in the above theorem arises when integrating 
the Dirac operator by parts. This allows us to interpret the term as a boundary contribution. 
Typically, such boundary terms describe the flux of Dirac currents through a given surface. In our 
case, however, due to the different masses of the wave functions, it is difficult to interpret this 
directly as a flux in the classical sense.

Nevertheless, in the limit $m \to m'$, the double integral can be interpreted as the probability 
flux of the Dirac current through the Cauchy and event horizons. We therefore define the {\em{fermionic flux operator}} as
\[
    \Scpr{\psi_m}{\mathfrak{B}_m \,\phi_m} := \lim_{m' \rightarrow m} \mathfrak{B}(\psi_m, \, 
    \phi_{m'}) \, ,
\]
where $\mathfrak{B}(\psi_m, \phi_{m'})$ is defined by the integrand of equation~\eqref
{massDecompReg}. We want to point out that this is the only natural definition of the fermionic flux 
operator. The reasoning follows that of~\cite{sigbh}, to which the interested reader is referred.
Additionally, in this limit, the contributions from the Cauchy horizon and the interior
event horizon cancel each other in the boundary matrix $\Omega(\omega, m)$, leaving its form 
similar to the exterior Schwarzschild case. 
This point will be discussed in Section~\ref{section5}.

In contrast to the fermionic flux term, the first summand in the
mass decomposition~\eqref{massDecompSing}
only involves a single mass integral.
This contribution describes the behavior of the Dirac
waves in the asymptotic region for large times.}} \QEDrem
\end{Remark}
With the above observations in mind, we can apply the mass decomposition theorem to introduce the 
fermionic signature operator and the fermionic flux operator in the Reissner-Nordström geometry. Since 
they have a joint spectral decomposition with the angular operator $\mathcal{A}$ (as discussed 
later) we define them for fixed angular mode $(k,l)$. For a better overview, we again add the wave 
mode indices.
\begin{Def}
    \label{def:sig_op_def}
    The fermionic signature operator $\S^{kl}_m$ in a Reissner--Nordström space time $(\mm, g)$ is 
    defined for some $\psi_m, \, \phi_m \, \in \, \hh^{\infty}$ as the map $\S^{kl}_m : \hh^
    {\infty} 
    \longrightarrow \hh^{\infty}$ given by
    \begin{align}
        \Scpr{\psi_m}{\S^{kl}_m \, \phi_m}_m := 4\pi \int_{|\omega|>m} \dif{\omega} \, \epsilon
        (\omega) \, \sum_{i,j=1}^2 \overline{\widehat{\psi}^{kl}_{i,m}(\omega)}\, \widehat{\phi}^
        {kl}_{j,m}(\omega) \, \braket{f^{kl\omega}_{i,m}}{f^{kl\omega}_{j,m}}_{\C^2}
    \end{align}
\end{Def}
\begin{Def}
    \label{def:flux_op}
    The fermionic flux operator $\mathfrak{B}^{kl}_m$in a Reissner-Nordström space time $(\mm, g)$ is 
    defined for some $\psi_m, \, \phi_m \, \in \, \hh^{\infty}$ as the map $\mathfrak{B}^{kl}_m : 
    \hh^{\infty} \longrightarrow \hh^{\infty}$ given by
    \begin{align}
        \Scpr{\psi_m}{\mathfrak{B}^{kl}_m \,\phi_m}_m := -2 \pi
        \int_{\R}\dif{\omega} \sum_{i,j = 1}^2 \overline{\widehat{\psi}^{kl}_{i,m}(\omega)}\, 
        \widehat{\phi}^{kl}_{j,m}(\omega)\, \Omega(\omega, m)_{ij}
    \end{align}
\end{Def}
\section{The Fermionic Signature Operator}\label{section4}
In this section, we compute the spectrum of the fermionic signature operator. In doing so, we want 
to analyze what information the spectrum encodes about the solution space for fixed $m$, which we 
denote by $\hh^\infty_m$. Using Definition~\ref{def:sig_op_def}, the result of Lemma~\ref
{scalLemma}, and Theorem~\ref{thm:massDecompTheo}, we obtain an analytic expression for the 
fermionic signature operator. This result is stated in the following proposition.
\begin{Prp}\label{prp:signatureOp}
    For fixed angular momentum modes $(k, l)$, the fermionic signature operator $\S_m$ in the 
    Reissner–Nordström geometry can be expressed for all $\psi \in \hh^{\infty}_m$ as follows:
    \begin{align}
        \S^{kl}_m = \dfrac{1}{\pi} \int_{|\omega| > m} \dif{\omega} \, \epsilon({\omega}) \sum_
        {i,j= 1}^2\, \big( T^2 \big)_{ij} \,\Scpr{\Psi^{kl}_{j,m}(\omega)}{\, \cdot} \Psi^{kl}_{i,m}(\omega, 
        r, \vartheta, \varphi)
    \end{align}
    Additionally, the integral kernel has the form
    \begin{align}\label{eq:kernle_s}
        \widehat{(\S\psi)}^{kl}_{i,m} =2 \, \epsilon(\omega)\, \Idmat_{|\omega| > m}(\omega)\, 
        \sum_{j = 1}^2 t_{ij} \,\widehat{\psi}^{kl}_{j,m}(\omega)\, .
    \end{align}
\end{Prp}
\begin{proof}
    Let $\phi_m , \psi_m \in \hh^{\infty}_m$. Comparing the expressions for the signature operator 
    from Definition~\ref{def:sig_op_def} with Lemma~\ref{scalLemma}, we obtain
    \begin{align}
        2\pi &\int_{\R}\dif{\omega} \sum_{i, \, j = 1}^2 (T^{-1})^{ij} \overline{\widehat{\psi}^
        {kl}_{j,m}(\omega)} \widehat{(\S \phi)}^{kl}_{i,m} \underbrace{=}_{\text{lemma}} \Scpr
        {\psi_m}{\S^{kl}_m \phi_m} \notag \\
        &\underbrace{=}_{\text{def}} 4\pi \int_{|\omega|>m} \dif{\omega} \, \epsilon({\omega}) 
        \sum_{i,j= 1}^2 \overline{\widehat{\psi}^{kl}_{i,m}(\omega)} \widehat{\phi}^{kl}_{j,m}
        (\omega) \braket{f^{kl\omega}_{i,m}}{f^{kl\omega}_{j,m}}_{\C^2}\:.
    \end{align}
    Using the identity from the inverse matrix $T^{-1}$ from Lemma~\ref{le:integralRep} and that 
    $f^{kl\omega}_{i,m}$ are $\C^2$-orthonormal vectors, the above equation can be written in the 
    form
    \begin{align}
        &\int_{|\omega| > m} \dif{\omega} \sum_{i,j = 1}^2 (T^{-1})^{ij}\,\overline
        {\widehat{\psi}^{kl}_{j,m}(\omega)}\,\widehat{(\S\phi)}^{kl}_{i,m}(\omega) \notag \\
        &=\int_{\R} \dif{\omega} \sum_{i,b=1}^2 (T^{-1})^{ib} \, \overline{\widehat
        {\psi}^{kl}_{i,m}(\omega)}\,2\, \Idmat_{|\omega|>m}(\omega)\, \epsilon({\omega}) \sum_{j=1}
        ^2 t_{bj} \, \widehat{\phi}^{kl}_{m, j}(\omega)\:,
    \end{align}
    from which the integral kernel can be read off. In the next step, we use the expression 
    for the kernel together with Lemma~\ref{le:integralRep} to obtain
    \begin{align}
        (\S^{kl}_m&\,\psi_m)(r, \vartheta, \varphi) = \int_{\R} \dif{\omega} \, \sum_{i=1}^2 
        \widehat{(\S \psi)}^{kl}_{i,m}(\omega) \Psi^{kl}_{i,m}(\omega, r, \vartheta, \varphi) 
        \notag \\
        & =2 \int_{|\omega| > m}\dif{\omega}\,\epsilon({\omega})\sum_{i, j =1}^2 t_{ij}\,
        \widehat{\psi}^{kl}_{j,m}(\omega) \Psi^{kl}_{i,m}(\omega, r, \vartheta, \varphi) \notag \\
        & = \dfrac{1}{\pi}\int_{|\omega| >m} \dif{\omega}\,\epsilon({\omega}) \sum_{i,j=1}^2 
        (T^2)_{ij} \, \Scpr{\Psi^{kl}_{j,m}(\omega)}{\psi_{m}}\big|_{\tau = 0} \Psi^{kl}_{i, m}
        (\omega, r, \vartheta, \varphi)\, .
    \end{align}
    This completes the proof.
\end{proof}
Now we compute the spectrum of the signature operator and show that it is bounded, so that $\S_m$ 
is a self-adjoint operator on $\hh^\infty_m$.
\begin{Corollary}\label{cor:props_s}
For fixed angular momentum modes $(k,l)$,
the fermionic signature operator $\S^{kl}_m$ is  a 
    symmetric, bounded operator with norm
    \begin{align}
        \|\S^{kl}_m \|_{\rm{op}} \leq 2 \:.
    \end{align}
    It commutes with the Dirac Hamiltonian $H$ and admits a spectral decomposition of 
    the form
    \begin{align}
        \S^{kl}_m = \int_{|\omega| > m} \, \S^{kl}_m(\omega) \, \dif{E}_{\omega} \:,
    \end{align}
    where $\dif{E}_{\omega}$ denotes the spectral measure of $H$. Moreover, 
    the operator~$\S^{kl}_m(\omega)$ has 
    the eigenvalues
    \begin{align}
        \lambda_\pm(\S^{kl}_m) = \left\{ \begin{array}{ll}
        2\, \epsilon(\omega)\,n_F (\omega) & \text{in case~$+$} \\
         2\, \epsilon(\omega)\,(1 - n_F (\omega)) & \text{in case~$-$} \end{array} \right.
        \quad \text{with} \quad n_F(\omega) = \dfrac{1}{1+e^{2\vartheta(\omega)}},
    \end{align}
    and
    \[
        \vartheta(\omega) = \dfrac{1}{2}\:{\rm{arccosh}}\big(\|f^\omega_{2,m}\|_{\C^2}\big) \:.
    \]
    with $f^\omega_{2,m}$ the transmission vector of $\Phi_2$ defined in~\eqref{eq:bc_phi2}.
\end{Corollary}
\begin{proof}
    Firstly, we notice after Proposition~\ref{prp:signatureOp} that the signature operator is a 
    multiplication operator in $\omega$ and thus commutes with the Hamiltonian. Furthermore, $\S^
    {kl}_m$ is a weighted spectral decomposition along the fundamental solutions $\Psi^{kl}_{m,i}$ 
    (more precisely, one can identify the scalar product $(\Psi^{kl}(\omega)_i | \cdot)$ as 
    spectral transformation). Using the Plancherel identity we get that $\S^{kl}_m$ is a 
    multiplication operator in its spectral representation. Therefore, when looking at the integral 
    kernel~\eqref{eq:kernle_s} we just need to compute the eigenvalues of the matrix $T$ from 
    Lemma~\ref{le:integralRep}
    \[
        T = \dfrac{1}{2}\begin{bmatrix}
            1 & e^{i(\beta - \gamma)}\tanh(\vartheta) \\
            e^{-i(\beta - \gamma)}\tanh(\vartheta) & 1
        \end{bmatrix} \, .
    \]
    It has the eigenvalues
    \begin{align}
        \lambda_{1,2} = \dfrac{1}{2}(1 \mp \tanh(\vartheta))\, ,
    \end{align}
    where $\vartheta(\omega)$ is defined in the discussion above the integral representation in 
    equation~\eqref{eq:vartheta}. Knowing, the eigenvalues we can compute the operator norm
    \begin{align}\label{boundNorm}
        \|\S^{kl}_m\|_{\rm{op}} = 2 \max\{|\lambda_+|, |\lambda_-|\} \leq 2 \: ,
    \end{align}
    which is uniformly in $k, l$ and $\omega$. Finally, evaluating the remaining part of the kernel 
    and using
    \[
        \dfrac{1}{2}\big(1-\tanh(\vartheta)\big) = \dfrac{1}{e^{2\vartheta} + 1} =: n_F(\vartheta)
    \]
    gives the desired result.
\end{proof}
Lastly, since $\S_m$ is uniformly bounded in $l$ and $\omega$ it has a joint spectral decomposition 
with the angular operator $\mathcal{A}$.
\begin{Corollary}\label{cor:spec_decomp_sig}
    The fermionic signature operator has a joint spectral decomposition with $\mathcal{A}$ and can be extended to an operator on~$\hh^\infty_m$,
    \begin{align}\label{eq:full_sig}
    \S_m = \sum_{k,l} \S_m^{kl} \, A_{k,l} : \hh^\infty_m \longrightarrow \hh^\infty_m \, ,
\end{align}
    where $S^{kl}_m$ is the fermionic signature operator for fixed angular modes $(k,l)$ from 
    Definition~\ref{def:sig_op_def} and $A_{k,l}$ the spectral measure from equation~\eqref
    {eq:specdecomp_angular}.
\end{Corollary}
\begin{proof}
    This follows directly from the fact, that the full Dirac propagator commutes with $\mathcal{A}$ 
    in a Reissner-Nordström geometry. Since, $\S_m$ has a spectral decomposition with $H$ we can 
    write
    \[
        \S_m = \sum_{k,l}\int_{|\omega| > m} \S^{kl}_m(\omega) \, \dif{E}_\omega \, A_{k,l} = \sum_
        {k,l} \S^{kl}_m \, A_{k,l} \, . 
    \]
    Because, we know from Corollary~\ref{cor:props_s} that $\|\S^{kl}_m\| \leq 2$, we get for the 
    full operator norm for a $\psi \in \hh_m$
    \[
        \Big\|\sum_{k,l} \S^{kl}_m A_{k,l} \psi\Big\|^2 \leq \sup_{k,l}\|\S^{kl}_m\|^2 \sum_{k,l}\|
        A_{k,l} \psi\|^2 \leq \|\psi\|^2 \, ,
    \]
    where we used in the last step the orthogonality of $A_{k,l}$.
\end{proof}
Note that this already follows abstractly from the fact that $\S_m$ respects the sym\-
metries of spacetime as worked out in detail in~\cite{sigsymm}.

\subsection{Construction of Quasi-free States}
Araki's construction of quasi-free states~\cite{araki1970quasifree} ensures that for any non-negative 
operator $W$ on $\hh^\infty_m$ satisfying $0 \leq W \leq 1$ and $W + F^* S F = 1$, where $F$ is an 
anti-unitary involution on $\hh^\infty_m$, there exists a unique quasi-free Dirac state whose 
two-point distribution is given by the integral kernel of $-W k_m$.

More generally, Araki's theorem can also be applied to arbitrary operators $W' = f(W)$, where $f$ 
is a non-negative Borel function. This allows one to construct an entire class of quasi-free states 
from $W$, each with different physical properties. In the following, we will use the fermionic 
signature operator to construct such quasi-free states. We emphasize that the
construction of this class of states is 
fully covariant, in the sense that the fermionic signature operator is defined in a manifestly 
covariant manner.

One possible state for the black hole is the fermionic projector state, which describes quasi-free 
Dirac particles; its definition, worked out in detail in~\cite{hadamard}, is recalled briefly 
below. By Corollary~\ref{cor:props_s}, $\S_m$ is a bounded symmetric operator on 
$\hh^\infty_m$, allowing us to define
\begin{Def}\label{projectorStates}
    The fermionic projector $P: C_0^\infty(\mm, S\mm) \longrightarrow \hh^\infty_m$ is given by
    \begin{align}
    P := -\Idmat_{(-\infty,0)}(\S_m)k_m,
    \end{align}
    where $k_m$ is the causal fundamental solution defined as
        \[
            k_m := \dfrac{1}{2\pi i}(s_m^\vee - s_m^\wedge) : C_0^\infty(\mm, S\mm) \longrightarrow 
            \hh_m^\infty \, .
        \]
\end{Def}
To obtain a Hadamard state, one requires a frequency splitting for an observer at infinity. This 
was demonstrated in~\cite{sigbh} for the exterior Schwarzschild geometry, and we extend the 
argument to the Reissner-Nordström geometry with horizon-penetrating coordinates.
\begin{Corollary}\label{projHad}
    The pure quasi-free fermionic projector state associated with the signature operator $\S_m$
	in the Reissner-Nordström geometry expressed in horizon-pene\-tra\-ting coordinates satisfies the 
    Hadamard condition.
\end{Corollary}
\begin{proof}
    By Corollary~\ref{cor:props_s}, the eigenvalues of $\S_m$ are negative if and only if 
    $\omega$ is negative. Therefore, the fermionic projector state coincides
    with the state obtained by frequency splitting for an observer at infinity.
    This quasi-free state is known to be of Hadamard form.
%    
%    Moreover, $\S_m$ vanishes for $\omega \in [-m,m]$, and by Corollary~\ref
%    {cor:ew_he} the eigenvalues for high energies coincide with those of the Minkowski Dirac 
%    Hamiltonian up to super-polynomial decaying error terms. Consequently, the negative-frequency 
%    subspace of $\S_m$ coincides with that of the Hamiltonian and differs only by smooth 
%    contributions in $\omega$:
%    \[
%        \Idmat_{(-\infty,-m]}(H) - \Idmat_{(-\infty,0)}(\S_m)  = \: (\text{smooth 
%        contributions})\, .
%    \]
%    Thus, the quasi-free fermionic projector state is Hadamard.
\end{proof}

\section{The Fermionic Flux Operator}\label{section5}
In this last section, we want to discuss the properties of the fermionic flux operator. As 
discussed in the previous Remark~\ref{rmk:flux}, $\mathfrak{B}(m,m')$ can only be interpreted as a 
probability flux through the horizons when performing the limit $m' \longrightarrow m$. In this 
setting, as an operator acting on elements of $\hh^\infty_m$, it describes the part of the wave 
that enters the black hole and disappears from the exterior causal region.

We start with deriving from the Definition~\ref{def:flux_op} the explicit form of the fermionic flux 
operator.
\begin{Prp}\label{prp:fluxOp}
    The fermionic flux operator $\mathfrak{B}^{kl}_m$ can be expressed for fixed angular momentum 
    modes $(k,l)$ and for all $\psi \in \hh^\infty_m$ as follows:
    \begin{align}\label{radialFLuxOp}
        \mathfrak{B}^{kl}_m = - \dfrac{1}{2\pi} \int_{|\omega|>m} \dif{\omega} \sum_{i,j=1}^2 \, 
        M_{ij}(\omega) \Scpr{\Psi^{kl}_{j,m}(\omega)}{\, \cdot}_m \, \Psi^{kl}_{m,i}(\omega,r, 
        \vartheta, \varphi),
    \end{align}
    and its kernel has the form
    \begin{align}\label{eq:kernel_b}
        (\mathfrak{B}\phi)^{kl}_{i,m} = -\Idmat_{|\omega| > m}(\omega) \sum_{a,b=1}^2 
        t_{i,a} s_a\,\widehat{\phi}^{kl}_{b,m}(\omega) \, ,
    \end{align}
    with 
    \[
        M^{kl}_{ij}(\omega) = \sum_{a}^1 t_{ia}\, s_a \,t_{aj} \, ,
    \]
    where $s_1 = 1$ and $s_2 = -1$.
\end{Prp}
\begin{proof}
    We start from Definition~\ref{def:flux_op}. Firstly, note that for $|\omega| < m$ only the 
    contributions with $i=j=1$ do not vanish, according to lemma~\ref{le:integralRep}. Secondly, 
    recall from the proof of Theorem~\ref{thm:massDecompTheo} that $\Omega(\omega,m)_{11} = 0$, as 
    stated in~\eqref{eq:zeroWronskian}. Hence, we obtain
    \[
        \Scpr{\psi_m}{\mathfrak{B}^{kl}_m\phi_m} := -2\pi \int_{|\omega| > m}\dif{\omega} \sum_{i,j 
        =1}^2 \overline{\widehat{\psi}^{kl}_{i,m} (\omega)}\, \widehat{\phi}^{kl}_{j,m}(\omega) \, 
        \Omega^{kl}(\omega,m)_{ij}.
    \]
    We proceed with the same steps as for the fermionic signature operator. Inserting the identity 
    relation, we obtain
    \begin{align}
        \Scpr{\psi_m}{\mathfrak{B}^{kl}_m\phi_m} := -2 \pi \int_{|\omega| > m}\dif{\omega} &\sum_{i,
        j = 1}^2(T^{-1})^{ij} \,\overline{\widehat{\psi}^{kl}_{i,m}(\omega)} \notag \\
        &\times \sum_{a,b=1}^2 t_{ja}\, \Omega^{kl}(\omega,m)_{ab}\, \widehat{\phi}^{kl}_{b,m}(\omega)
    \end{align}
    Comparing this relation with Lemma~\ref{scalLemma} yields the integral kernel of the fermionic 
    flux operator. Using this kernel together with Lemma~\ref{le:integralRep}, we find
    \begin{align}
        (\mathfrak{B}&\phi)^{kl}_m(r,\, \vartheta, \, \varphi) = \sum_{a=1}^2 (\mathfrak{B}\phi)^
        {kl}_{a,m}(\omega) \Psi^{kl}_{a,m}(\omega, r, \vartheta, \varphi) \notag \\
        &= - \dfrac{1}{2\pi}\int_{|\omega| > m} \dif{\omega}   \sum_{i,a,b,j} t_{ia} \,\Omega^{kl}(\omega,m)_{ab}\, t_{bj} \Scpr{\Psi^{kl}_{j,m}(\omega)}
        {\phi_m}_m \Psi^{kl}_{i,m}(\omega, r,\vartheta,\varphi)\notag \\
        &=- \dfrac{1}{2\pi} \int_{|\omega|>m} \dif{\omega} \sum_{i,j=1}^2 \,M^{kl}_{ij}(\omega) 
        \Scpr{\Psi^{kl}_{j,m}(\omega)}{\psi_m}_m\,\Psi^{kl}_{m,i}(\omega, r, \vartheta, \varphi)\, ,
    \end{align}
    with 
    \[
        M^{kl}_{ij}(\omega) = \sum_{a,b=1}^1 t_{ia}\, \Omega^{kl}(\omega,m)_{ab} \, t_{bj} \, .
    \]
    Using the fact that $\Omega^{kl}(\omega,m)$ equals the Wronskian at the exterior event horizon 
    and that the $f^{kl\omega}_{i,m}$ are orthogonal, completes the proof.
\end{proof}
Next, we derive the explicit form of the eigenvalues and other properties of $\mathfrak{B}^{kl}_m$. 
We want to highlight that these results depend on the evolution of the wave function in the 
whole interval $(r_-, \infty)$ and takes contributions from the interior of the black hole into 
account.
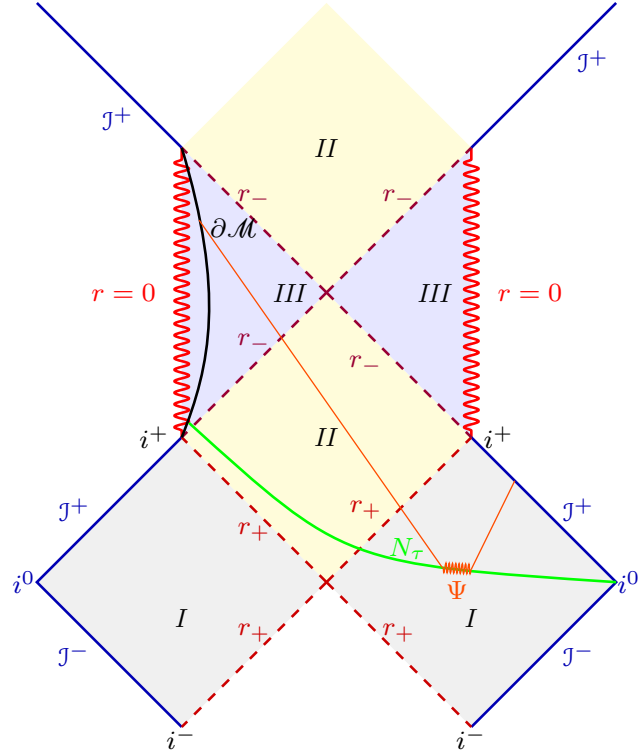
\begin{figure}[t]
    \centering
    \resizebox{0.6\textwidth}{!}{
    \begin{tikzpicture}[scale=1.55, font=\small]
        \begin{scope}[shift={(0,0)}]  % reset

            % Region I  (exterior right)
            \coordinate (cI)    at ( 1, 0);
            \coordinate (I-r)   at ( 2, 0);   % i^0
            \coordinate (I-t)   at ( 1, 1);   % future on r_+
            \coordinate (I-l)   at ( 0, 0);   % past junction
            \coordinate (I-b)   at ( 1,-1);   % i^- side  (only partially shown)

            % Region IV (exterior left, mirror)
            \coordinate (cIV)   at (-1, 0);
            \coordinate (IV-r)  at ( 0, 0);
            \coordinate (IV-t)  at (-1, 1);
            \coordinate (IV-l)  at (-2, 0);
            \coordinate (IV-b)  at (-1,-1);

            % Region II (between horizons)
            \coordinate (cII)   at ( 0, 1);
            \coordinate (II-r)  at ( 1, 1);
            \coordinate (II-t)  at ( 0, 2);
            \coordinate (II-l)  at (-1, 1);
            \coordinate (II-b)  at ( 0, 0);

            % Region II (between horizons, after interior)
            \coordinate (cVI)   at ( 0, 3);
            \coordinate (VI-r)  at ( 1, 3);
            \coordinate (VI-t)  at ( 0, 4);
            \coordinate (VI-l)  at (-1, 3);
            \coordinate (VI-b)  at ( 0, 2);

            % Region III (inner, r < r_-) (triangle right)
            \coordinate (cIII)  at ( 0.5, 2);
            \coordinate (III-r) at ( 1, 3);
            \coordinate (III-t) at ( 1, 3);
            \coordinate (III-l) at ( 0, 2);
            \coordinate (III-b) at ( 1, 1);

            % Region III (inner, r < r_-) (triangle left, mirror)
            \coordinate (cV)  at ( -0.5, 2);
            \coordinate (V-t) at ( -1, 3);
            \coordinate (V-r) at ( 0, 2);
            \coordinate (V-b) at ( -1, 1);

            % Special curve points
            \coordinate (int-p) at (0.5, 0.5);
            \coordinate (s-cauchy) at (1, 0.09);
            \coordinate (e-cauchy) at (0.8, 0.1);

            % ==================================================================
            %  FILLS  (paint regions before drawing lines)
            % ==================================================================

            % Region I
            \fill[fill-I]
            (I-r) -- (I-t) -- (I-l) -- (I-b) -- cycle;

            % Region IV (mirror)
            \fill[fill-IV]
            (IV-r) -- (IV-t) -- (IV-l) -- (IV-b) -- cycle;

            % Region II
            \fill[fill-II]
            (II-r) -- (II-t) -- (II-l) -- (II-b) -- cycle;

            % Region VI (mirror)
            \fill[fill-II]
            (VI-r) -- (VI-t) -- (VI-l) -- (VI-b) -- cycle;

            % Region III
            \fill[fill-III]
            (III-t) -- (III-l) -- (III-b) -- cycle;

            % Region V (mirror)
            \fill[fill-III]
            (V-t) -- (V-r) -- (V-b) -- cycle;

            % ==================================================================
            %  SINGULARITIES  r = 0 
            % ==================================================================

            \draw[singularity] (1, 1) -- (1, 3)
            node[right=4pt, midway, red, font=\scriptsize] {$r=0$};
            \draw[singularity] (-1, 1) -- (-1, 3);
            \node[left=4pt, red, font=\scriptsize] at (-1, 2) {$r=0$};

            % ==================================================================
            %  CONFORMAL INFINITY  (scri+, scri-, i0, i+, i-)
            % ==================================================================

            % Right exterior I\begin{document}
            \draw[scri] (I-b) -- (I-r) node[right, inner sep=0pt,
                font=\scriptsize] {$i^0$} -- (I-t);
            % Label scri
            \node[blue!70!black, right=2pt,
                font=\scriptsize] at ($(I-b)!0.5!(I-r)$) {$\mathscr{I}^-$};
            \node[blue!70!black, right=2pt,
                font=\scriptsize] at ($(I-r)!0.5!(I-t)$) {$\mathscr{I}^+$};

            % Left exterior IV
            \draw[scri] (IV-b) -- (IV-l) node[left, inner sep=0pt,
                font=\scriptsize] {$i^0$} -- (IV-t);
            \node[blue!70!black, left=2pt,
                font=\scriptsize] at ($(IV-b)!0.5!(IV-l)$) {$\mathscr{I}^-$};
            \node[blue!70!black, left=2pt,
                font=\scriptsize] at ($(IV-l)!0.5!(IV-t)$) {$\mathscr{I}^+$};

            % Upper sheet I'
            \draw[scri] (III-r) -- (2, 4);
            \draw[scri] (V-t) -- (-2, 4);
            \node[blue!70!black, right=2pt, 
                font=\scriptsize] at ($(III-r)!0.6!(2,4)$) {$\mathscr{I}^+$};
            \node[blue!70!black, left=2pt,
                font=\scriptsize]  at ($(III-l)!0.6!(-2,4)$) {$\mathscr{I}^+$};

            % i± points
            \node[below, inner sep=0pt, font=\scriptsize] at (I-b)  {$i^-$};
            \node[right, font=\scriptsize] at (I-t)  {$i^+$};
            \node[below, inner sep=0pt, font=\scriptsize]  at (IV-b) {$i^-$};
            \node[left, font=\scriptsize]  at (IV-t) {$i^+$};

            % ==================================================================
            %  EVENT HORIZON  r = r_+
            % ==================================================================

            % Future event horizon H+: the two null lines bounding region II from below
            \draw[horizon]
            (II-b) -- (II-r)
            node[midway, below,left=1pt, font=\scriptsize] {$r_+$};
            \draw[horizon]
            (I-l) -- (I-b)
            node[midway, above, font=\scriptsize] {$r_+$};
            \draw[horizon]
            (II-b) -- (II-l)
            node[midway, below, font=\scriptsize] {$r_+$};
            \draw[horizon]
            (IV-r) -- (IV-b)
            node[midway, above, font=\scriptsize] {$r_+$};

            % ==================================================================
            %  CAUCHY HORIZON  r = r_-
            % ==================================================================

            \draw[cauchy-hor] (II-t) -- (III-r)
            node[midway, above, font=\scriptsize] {$r_-$};
            \draw[cauchy-hor]
            (II-t) -- (II-l)
            node[midway, above, font=\scriptsize] {$r_-$};
            \draw[cauchy-hor]
            (V-r) -- (V-t)
            node[midway, above, font=\scriptsize] {$r_-$};
            \draw[cauchy-hor]
            (III-l) -- (III-b)
            node[midway, below, left=1pt, font=\scriptsize] {$r_-$};

            % ==================================================================
            %  REGION LABELS
            % ==================================================================
            % Region labels (overwrite the duplicate gray ones)
            \node[font=\scriptsize] at ( 1,  -0.25)   {\textit{I}};
            \node[font=\scriptsize] at (-1,  -0.25)   {\textit{I}};
            \node[font=\scriptsize] at ( 0,    1)   {\textit{II}};
            \node[font=\scriptsize] at ( 0,    3)   {\textit{II}};
            \node[font=\scriptsize] at ( 0.75,  2)   {\textit{III}};
            \node[font=\scriptsize] at (-0.25,  2)   {\textit{III}};

            % ==================================================================
            % Curves
            % ==================================================================
            % Boundary around singularity
            \draw[thick, black,
                decoration={amplitude=1pt, segment length=5pt},
                decorate]
            (-1, 1)
            .. controls (-0.8, 1.5) and (-0.7, 2) .. 
            (-1, 3)
            node[midway, above=28pt, right, inner sep=0pt, font=\scriptsize] {$\partial \mm$};;

            % Cauchy surface
            \draw[thick, green,
                decoration={amplitude=1pt, segment length=5pt},
                decorate]
            (I-r)
            .. controls (0.125,0.125) .. 
            (-0.95,1.1)
            node[midway, right=5pt, font=\scriptsize]{$N_\tau$};

            % Ineraction Cauchy surface
            \draw[em-wave] (s-cauchy) -- (e-cauchy)
            node[midway, below, font=\scriptsize] {$\Psi$};     % interacting onto Cauchy surface
            \draw[orange!60!red] (s-cauchy) -- (1.3, 0.7);      % one part to future inf
            \draw[orange!60!red] (e-cauchy) -- (-0.88, 2.5);    % other part through CH 
        \end{scope}
    \end{tikzpicture}
    }
    \caption{Maximal analytical extension of a Reissner-Nordström black hole. The green line 
    indicates the constructed Cauchy hypersurfaces $N_\tau$. The surface $\partial \mm$ denotes the 
    boundary constructed to obtain an essential self adjoint Hamiltonian for the spectral analysis. 
    $\Psi$ (orange) denotes an interacting fermionic wave propagating from the Cauchy surface. The 
    event horizon (dark red) lies at $r_+$ and the Cauchy horizon (purple) at $r_-$. Massive waves 
    cannot travel faster than light and therefore have a slope strictly greater $45^\circ$.}
    \label{fig:penrose_diag}
\end{figure}
\begin{Prp}
    The fermionic flux operator $\mathfrak{B}^{kl}_m$, acting on $\hh^\infty_m$ for fixed angular 
    momentum modes $(k,l)$ is a symmetric and bounded operator with operator norm
    \begin{align}
        \|\mathfrak{B}^{kl}_m\|_{\rm{op}} \leq 1
    \end{align}
    and eigenvalues
    \begin{align}
        \lambda_\pm(\mathfrak{B}^{kl}_m) = \mp \sqrt{2 n_F(\vartheta)} \in \R \, ,
    \end{align}
    with $n_F$ the distribution function from the eigenvalues of $\S^{kl}_m$ from Corollary~\ref
    {cor:props_s}.Furthermore, it commutes with the Dirac Hamiltonian $H$ and hence admits a 
    spectral decomposition of the form
    \[
        \mathfrak{B}^{kl}_m = \int_{|\omega| > m} \mathfrak{B}^{kl}_m(\omega) \dif{E}_{\omega} \, .
    \]
\end{Prp}
\begin{proof}
    Looking at the representation of $\mathfrak{B}^{kl}_m$ in Proposition~\ref{prp:fluxOp}, we see 
    that it acts as a multiplication operator in $\omega$ (we can interpret this as another 
    spectral decomposition along the fundamental solutions, as for the signature operator). 
    Therefore, it commutes with $H$. 
    
    To obtain the eigenvalues, we need to diagonalize the operator in~\eqref{eq:kernel_b}. A short 
    computation gives
    \[
        \lambda(\mathfrak{B}^{kl}_m)_\pm = \pm \sqrt{1 - \tanh(\vartheta)} \, .
    \]
    Using the maximal eigenvalue for the operator norm gives the result for the uniform bound.
\end{proof}
Next, because $\mathfrak{B}^{kl}_m$ is uniformly bounded in $l$ and $\omega$, it as a joint 
spectral decomposition with $\mathcal{A}$.
\begin{Thm}
    The fermionic flux operator admits a joint spectral decomposition with $\mathcal{A}$ and acts 
    on the full $\hh_m^\infty$:
    \[
        \mathfrak{B}_m = \sum_{k,l}\mathfrak{B}_m^{kl}\, A_{k,l} : \hh^\infty_m \longrightarrow 
        \hh^\infty_m \, ,
    \]
    with $A_{k,l}$ the spectral measure from equation~\eqref{eq:specdecomp_angular}.
\end{Thm}
\begin{proof}
    The arguments are identical as in the proof of Corollary~\ref{cor:spec_decomp_sig}.
\end{proof}
As already mentioned in Remark~\ref{rmk:flux}, the fermionic flux operator does not depend on the flux 
through the Cauchy horizon explicitly, even though the non-singular part of the mass decomposition 
is a function of the full boundary matrix as stated in Theorem~\ref{thm:massDecompTheo}. This is 
caused by the radial independence of the Wronskian for $m' \rightarrow m$, which cancels the 
interior contributions from both horizons.

Physically, the fermionic flux operator describes the 
flux of the Dirac current through the event horizon.
Note that, after crossing the event horzion, a Dirac wave necessarily also crosses
the Cauchy horizon. The full propagation of a Dirac particle $\Psi$ is shown in 
Figure~\ref{fig:penrose_diag}. After 
interacting on the Cauchy surface (green) $N_\tau$, one part of the wave crosses the 
event horizon 
at $r_+$ and travels all the way up to the boundary $\partial \mm$, which shields the 
singularity and 
reflects the wave into the analytic extension. This part encodes the fermionic flux 
operator $\mathfrak
{B}_m$.The other part of the Dirac wave propagates to future infinity, where the particle 
can be detected.

We finally point out that both the fermionic signature operator~$\S_m$
and the fermionic flux operator~$\mathfrak{B}_m$ depend on the global
geometry of spacetime. We saw that their eigenvalues depend on
whether we consider the spacetime up to the Cauchy horizon or
whether we consider only the exterior geometry. This means in particular that,
if an observer at infinity could determine the eigenvalues of~$\S_m$
or~$\mathfrak{B}_m$, he could gain information on the interior geometry
of the black hole.
%%%%%%%%%%%%%%%%%%%%%%%%%%%%%%%%%%%%%%%%%%%%%%%%%%%%%%%%%%%%%%%%%%%%%%%%%%%%%%%%%%%%%%%%%%%%%%%%%%%%
%%%%%%%%%%%%%%%%%%%%%%%%%%%%%%%%%%%%%%%%%%%
%%%%%%%%%%%% Quellen %%%%%%%%%%%%%%%%%%%%%%%%%%%%%%%%%%%%%%%

\Thanks{
    {{\rm{Acknowledgments: C.K.\ gratefully acknowledges support by the 
    Heinrich-B\"oll-Stiftung.}}}
}

\bibliographystyle{amsplain}
%\bibliography{./../quellen}
\providecommand{\bysame}{\leavevmode\hbox to3em{\hrulefill}\thinspace}
\providecommand{\MR}{\relax\ifhmode\unskip\space\fi MR }
\providecommand{\MRhref}[2]{%
  \href{http://www.ams.org/mathscinet-getitem?mr=#1}{#2}
}
\providecommand{\href}[2]{#2}

\end{document}